%% file: main_arxiv.tex
\documentclass{article}
\usepackage{amsmath} 
\usepackage{nccmath}
\usepackage{multirow}
\usepackage{amssymb}  
\usepackage{mathtools}
\usepackage[english]{babel}
\usepackage{stmaryrd}
\usepackage{nomencl}

\makenomenclature

\newtheorem{theorem}{Theorem}

\newtheorem{lemma}{Lemma}

\newtheorem{assumption}{Assumption}
\input{head.tex}
\input{style.tex}
\input{defs-math-power.tex}

\usepackage[RPvoltages,siunitx]{circuitikz}
\usepackage{standalone}
\usepackage{empheq}
\usepackage[theorems]{tcolorbox} 
\usepackage[final]{microtype} 
\tcbset{highlight math style={notitle,nophantom,colframe=black,colback=black!5!white}}

\title{Phase Retrieval via Model-Free Power Flow Jacobian Recovery\thanks{This material is based upon work supported by the National Science Foundation Graduate Research
Fellowship Program under Grant No. DGE-1650044. Any opinions, findings, and conclusions
or recommendations expressed in this material are those of the author(s) and do not necessarily reflect the views of the National Science Foundation.\\
\hspace*{1.8em}The authors are with the School of Electrical and Computer Engineering, Georgia Institute of Technology, Atlanta, GA, USA. \\
Email: \href{mailto:talkington@gatech.edu}{talkington@gatech.edu} 
\href{mailto:sgrijalva@ece.gatech.edu}{sgrijalva@ece.gatech.edu} Web: \url{https://talkington.dev}
}}
\author{Samuel Talkington, Santiago Grijalva
}

\begin{document}
\maketitle
\begin{abstract}
Phase retrieval is a prevalent problem in digital signal processing and experimental physics that consists of estimating a complex signal from magnitude measurements. This paper expands the classical phase retrieval framework to electric power systems with \emph{unknown} network models and limited access to observations of voltage \emph{magnitudes}, active power injections, and reactive power injections. The proposed method recovers the phase angles and the power-phase angle submatrices of the AC power flow Jacobian matrix. This is made possible by deriving topology and parameter-free expressions for the structural symmetries of the power flow Jacobian that do not depend on the phase angles. These physical laws provide structural constraints for the proposed phase retrieval method. The paper then presents sufficient conditions for guaranteed recovery of the voltage phase angles, which also depend solely on voltage magnitudes, active power injections, and reactive power injections. The method offers two significant benefits: both estimating the voltage phase angles and recovering the power flow Jacobian matrix\textemdash a basis for approximating the power flow equations. Simulations on widely studied open-source test networks validate the findings.
\end{abstract}

\section{Introduction}
\label{sec:introduction}
Flows of active and reactive power in AC electric power systems are directly related to the nodal voltage phase angles. Thus, there is great interest in knowing the phase angles for many network planning, control, and protection applications. To this end, broad efforts have taken place to deploy phasor measurement units (PMUs), which use communication networks and protocols to synchronize measurements of phase angles. 

Nonetheless, due to their complexity and cost, PMU deployment remains limited in some regions and grid types, such as in transmission system boundaries, transmission systems in rural environments, and distribution systems. Another method to determine the voltage phase angles is through state estimation. However, conventional state estimation techniques require detailed and complete knowledge of the network topology and model parameters. This is rarely possible in practice \cite{yuan_inverse_2022,deka_structure_learning}, particularly for combined transmission and distribution systems. In addition, network models are often inaccurate and difficult to maintain.

Therefore, we propose a method that enables the recovery of the voltage phase angles in electric power systems where the network model is unknown or unavailable. We extend previous works on recovering sensitivity matrices that relate the bus complex injections to bus voltage phasors to show that the structure of the AC power flow Jacobian enables the recovery of the power-to-phase angle sensitivities\textemdash and consequently, the phase angles\textemdash when neither observations of the phase angles nor the network topology are available.

\subsection{Problem description}
Each bus $i$ in an electricity network is modeled with a complex power injection and voltage phasor, denoted as 
\begin{equation}
\label{eq:pmu_signal_def}
     v_i \angle{\theta_i} \in \C, \quad 
    p_i + j q_i \in \C, \quad i=1,\dots,n. \tag{Phasors}
\end{equation}
\textcolor{black}{We often lack full information about these complex variables, but receive measurements of the form}
\begin{equation}
\label{eq:ami_signal_definition}
    \quad  v_i \in \R, \quad p_i + jq_i \in \C, \quad i =1,\dots,n, \tag{Magnitudes}
\end{equation}
where $v_i, p_i$,  and $q_i$ are the voltage \textit{magnitude}, active power injection, and reactive power injection at bus $i$, respectively. The measurements \eqref{eq:ami_signal_definition} may emerge from substation sensors in the context of transmission systems, or from sensors in ``advanced metering infrastructure" (AMI) in the context of distribution systems. Critically, these measurements lack information about the voltage phase angles $\theta_1,\dots,\theta_n$.

Motivated by the contrast between the measurements in \eqref{eq:ami_signal_definition} and \eqref{eq:pmu_signal_def}, we extend the classical \emph{phase retrieval} problem to electric power systems with unknown network models, through the lens of a recovered AC power flow Jacobian matrix  \cite{chen_measurement-based_2016,moffat_power_voltage_sensitivity}. Phase retrieval is a rich digital signal processing theory \cite{Gerchberg1972APA,fienup_phase_1982,burvall_phase_2011,shechtman_phase_2015,waldspurger_phase_2015,jaganathan_phase_2016,bahmani_phase_2017,goldstein_phasemax_2018-1,DONG-JONATHAN-PHRET-REVIEW} that consists of estimating a complex-valued signal $\vx \in \C^n$ given observations of real-valued \emph{magnitude} measurements $\vb \in \R^m$. In the case of a linear system of equations described by a matrix $\mA \in \C^{m \times n}$, classical phase retrieval is concerned with the problem
\begin{equation}
\label{eq:high-level-phir}
    \operatorname{find} \quad \vx \in \C^{n} \quad \text{subject to:} \quad | \mA \vx | = \vb,
\end{equation}
where $\vb \in \R^m$ are the element-wise magnitudes of the measurements $\mA \vx$.

Importantly, we will show that the structure of our particular problem allows us to circumvent full knowledge of the analogous $\mA$\textemdash i.e., the power flow Jacobian matrix\textemdash by exploiting the symmetries in its structure that arise due to electrical circuit physics. In Section \ref{sec:nrpf-review} and in Section \ref{sec:jacobian_structure_review}, this structure is reviewed. We derive a novel representation of the symmetries within this structure in Section \ref{sec:phaseless_symmetry}, which then leads to the recovery of the  voltage phasors and the full power flow Jacobian in Section \ref{sec:phasecut-nrpf}. This representation depends on neither the phase angles nor the network model. Moreover, we then use this representation to develop sufficient conditions for guaranteed voltage phase retrieval and power flow Jacobian invertibility in Section \ref{sec:sufficient-cond-unique}.

\subsection{Contributions}

In summary, this work proposes a novel application of phase retrieval\textemdash electric power systems\textemdash and synthesizes theory from both fields for this application. Specifically, considering only observations of the voltage magnitudes, active power injections, and reactive power injections as in \eqref{eq:ami_signal_definition}, we show that it is possible to recover the voltage phasors as in \eqref{eq:pmu_signal_def}.  This is achieved by deriving and exploiting a novel expression for the physical law describing symmetries in the power flow Jacobian matrix. The proposed method  contributes solutions to several classes of power system estimation problems from \emph{only} the measurements in \eqref{eq:ami_signal_definition}. Specifically, the contributions of this paper are:
\begin{enumerate}
    \item Estimating the voltage phase angles $\vtheta$ as in \eqref{eq:pmu_signal_def}, or their perturbations $\Delta \vtheta$, from \emph{only} the measurements in \eqref{eq:ami_signal_definition}, and developing sufficient conditions that guarantee voltage phase retrieval.
    \item Estimating the power-\textit{voltage phase angle} sensitivity matrices $\DPDTH, \DQDTH \in \mathbb{R}^{n \times n}$, which are  two of the four blocks of the \textit{power flow Jacobian} $\mJ(\vx) \in \R^{2n \times 2n}$, which is defined for a given operating point of the network $\vx \triangleq[ \vtheta^T,\vv^T ]^T \in \mathbb{R}^{2n}$.
    \item Estimating the current injection phase angles from observations of the current magnitudes described by the complex power measurements in \eqref{eq:ami_signal_definition}.
\end{enumerate}


\nomenclature[01]{\(j\)}{The imaginary unit ($j^2 = -1$)}
\nomenclature[02]{\(A,a\) }{Scalar}
\nomenclature[03]{\(\va\)}{Vector}
\nomenclature[04]{\(\mA\)}{Matrix}
\nomenclature[05]{\( \setA \)}{Set}
\nomenclature[06]{\(\mathds{1}\)}{Vector of all ones}
\nomenclature[07]{\(\mId\)}{The identity matrix}
\nomenclature[08]{\(\frac{\partial \vf}{\partial \vx}\)}{$m \times n$ Jacobian matrix of $\vf: \R^n \mapsto \R^m$}
\nomenclature[09]{\(\frac{\partial f_i}{\partial x_k}\)}{$(i,k)$-th entry of the Jacobian matrix $\frac{\partial \vf}{\partial \vx}$}
\nomenclature[10]{\(\norm{\va}_2\)}{Euclidean/$\ell_2$ norm of vector $\va$, $\norm{\va}_2 = \left(\sum_{i=1}^n a_i^2 \right)^{\frac{1}{2}}$}
\nomenclature[11]{\(\norm{\mA}_2\)}{Spectral norm/2-norm of matrix $\mA$, $\norm{\mA}_2 = \max_{\vx \neq \boldsymbol{0}} \frac{\norm{\mA \vx}_2}{\norm{\vx}_2}$}
\nomenclature[12]{\(\lambda(\mA)\)}{Eigenvalues of matrix $\mA$}
\nomenclature[13]{\(\sigma(\mA)\)}{Singular values of matrix $\mA$}
\nomenclature[14]{\(\diag(\vx)\)}{Diagonal matrix, with entries of vector $\vx$ along the diagonal}
\nomenclature[15]{\(\trace{\mA}\)}{Sum of the diagonal entries of matrix $\mA$}
\nomenclature[16]{\(\rank(\mA)\)}{Rank of matrix $\mA$}
\nomenclature[17]{$\lvert \cdot \rvert$}{Element-wise absolute value}
\nomenclature[18]{\((\cdot)^T\)}{Transpose}
\nomenclature[19]{\((\cdot)^*\)}{Complex conjugate transpose}
\nomenclature[20]{\((\cdot)^{\dagger}\)}{Moore-Penrose pseudoinverse}
\nomenclature[21]{\((\cdot) \circ (\cdot)\)}{Element-wise multiplication}
\nomenclature[22]{\(\conj{(\cdot)}\)}{Complex conjugate of a complex argument}
\nomenclature[23]{\(\Re{\cdot}\)}{Real component of a complex argument}
\nomenclature[24]{\(\Im{\cdot}\)}{Imaginary component of a complex argument}

{\footnotesize\printnomenclature}

\section{Preliminaries} 
\label{sec:prelim}

The benefits of having voltage phase information as in \eqref{eq:pmu_signal_def} are well-known \cite{terzija_wide-area_2011,liu_d-pmu_2020}. Knowledge of voltage phase information allows for improved power system model validation, control, and  protection \cite{terzija_wide-area_2011}. It also enables improvements in grid reliability, planning, and investment decision-making. Some of these benefits include increasing the accuracy of state estimators in distribution \cite{liu_trade-offs_2012,prasad_trade-offs_2018} and transmission \cite{zhou_alternative_2006} systems. However, the cost and complexity of PMUs have resulted in limited deployment. Approximately 3000 PMUs have been deployed in the United States as of 2022 \cite{kezunovic_2022_pmu_deplyoment_testing_lifecycle}, up from 2000 PMUs as of 2015 and 500 PMUs as of 2009 \cite{overholt_synchrophasor_2015}.

\subsection{Phaseless electricity network modeling techniques}
\label{sec:review-phase-angle-free-methods}
    Motivated by the heterogeneous deployment of PMUs, a growing literature has developed on \emph{phaseless}\textemdash or ``non-PMU"\textemdash measurement-based modeling techniques for electricity networks. This literature focuses on power systems where the network model is unknown, and where measurements are limited to $\vv,\vp,\vq$ as in \eqref{eq:ami_signal_definition}. 
    
    This literature has developed methods for estimating the unknown topology and electrical parameters\textemdash both for single-phase\cite{zhang_phaseless_topology_line_param} and multiphase unbalanced \cite{claeys_unbalanced_phaseless_line_param,chang_bus_impedance_matrix_estimation} networks. The phaseless estimation of the bus impedance matrix \cite{chang_bus_impedance_matrix_estimation} has recently seen experimental validation in hardware. 
    
    Additionally, this literature has developed methods for estimating complex power injections using voltage magnitude measurements. These methods largely center on the voltage-power submatrices of the inverse power flow Jacobian \cite{talkington_conditions_sensitivities,gupta_model-less_2022,christakou_efficient_2013}, and have also seen recent experimental validation in hardware \cite{chang_voltage_sensitivities_realworld,gupta2023experimental}.

    Our work contributes to this burgeoning literature. It provides the first set of methods and sufficient conditions for the recovering the voltage phase angles as in \eqref{eq:pmu_signal_def} given \emph{only} measurements of the form \eqref{eq:ami_signal_definition}\textemdash even when the network model is unknown. 

    \subsection{Review of the phase retrieval problem}
    \label{sec:phase-retrieval-review}
        Optimization problems of the form \eqref{eq:high-level-phir} are known as \emph{phase retrieval} problems. Classically, the goal of phase retrieval is to recover the phase of a complex signal $\vx \in \C^n$ \textcolor{black}{given a design matrix $\mA \in \C^{m \times n}$ and} real-valued magnitude measurements  $|\mA \vx | = \vb$, where $\vb \in \R^{m}$, and $|\cdot|$ denotes the elementwise magnitudes of a vector. In this paper, we make the noteworthy distinction that the design matrix (the power flow Jacobian) is not assumed to be known, in contrast with the classical approach to phase retrieval. 
        
        One of the earliest investigations of this problem was in experimental physics, where \cite{Gerchberg1972APA} developed an iterative method for recovering the phase of a wave function from the magnitude of its Fourier Transform coefficients, provided by intensity measurements. The authors of \cite{Gerchberg1972APA} note wide applications for their method in X-ray crystallography, which has been discussed more recently in \cite{burvall_phase_2011}. Since then, the phase retrieval problem and its various formulations have seen applications in electron microscopy \cite{fienup_phase_1982}; optical and medical imaging \cite{burvall_phase_2011}; experimental physics \cite{shechtman_phase_2015,jaganathan_phase_2016}; and digital signal processing. We refer the interested reader to \cite{jaganathan_phase_2016} and \cite{DONG-JONATHAN-PHRET-REVIEW} for multidisciplinary surveys of the formulations and applications of phase retrieval.

        \subsubsection{Classical phase retrieval}
            \label{sec:phase_retrieval}
            The classical phase retrieval problem, \textcolor{black}{where $\mA$ is assumed to be known, (in our case, it will be recovered)} can be understood as a quadratic optimization problem with complex decision variables:
            \begin{equation}
            \label{eq:basic_ph_retr_prob}
                \minimize_{\vx \in \C^n, \vy \in \C^m} \quad \norm{\mA \vx - \vy}_2^2 
                \quad \text{subject to:} 
                \quad |\vy| = \vb. 
            \end{equation}
            The introduction of the decision variables $\vy \in \C^m$ in \eqref{eq:basic_ph_retr_prob} is sometimes referred to as ``lifting" the magnitude measurements $\vb \in \R^m$ to $\C^m$, which allows us to recover the complex vector $\mA \vx \in \C^m$. It is noted in \cite{waldspurger_phase_2015} that the problem \eqref{eq:basic_ph_retr_prob} can be recast as solving for the complex signal $\vx \in \C^n$ and for the phases of the magnitude measurements $\vb$. The phases are modeled as rectangular complex numbers $\vu \in \C^m$ that are constrained to lie on the unit circle, i.e., $|u_i| = 1 \ \ \forall i=1,\dots,n$. With this framework, the phase retrieval program \eqref{eq:basic_ph_retr_prob} can be rewritten as
            \begin{equation}
            \label{eq:basic-ph-retr-prob-fixed-phase}
                \minimize_{\vx \in \C^n, \vu \in \C^m} \quad \norm{\mA \vx - \diag(\vb) \vu}^2_2 \quad \text{subject to:} \quad |\vu| = \mathds{1}.
            \end{equation}
            The authors of \cite{waldspurger_phase_2015} note that, for any fixed guess of the phase $\hat{\vu} \in \C^m$ of the magnitude measurements $\vb$, an estimate $\hat{\vx} \in \C^n$ of the complex signal is unique and is given as 
            \begin{equation}
                \hat{\vx} = \mA^{\dagger} \diag(\vb) \hat{\vu} = (\mA^* \mA)^{-1} \mA^* \diag(\vb) \hat{\vu},
            \end{equation}
            where $(\cdot)^*$ is the complex conjugate transpose and $(\cdot)^{\dagger}$ is the Moore-Penrose pseudoinverse.
            
            With this in hand, we can substitute $\vx$ with $\mA^{\dagger} \diag(\vb) \vu$, and solve the quadratic program
            \begin{equation}
             \label{eq:generalized-high-level-phir}
            \minimize_{\vu \in \C^m : |\vu| = \mathds{1}} \vu^* \mM \vu, \ \ \text{subject to:} \ \   \mM = \diag(\vb)(\mId - \mA \mA^{\dagger})\diag(\vb),
            \end{equation}
            where $\mM \succeq 0$. However, the modulus constraints in \eqref{eq:generalized-high-level-phir} are computationally challenging. They are equivalent to the quadratic equality constraints $\Re{u_i}^2 + \Im{u_i}^2 =1$ for all $i$, which produces a non-convex feasible region.

\subsubsection{The PhaseCut algorithm}
            \label{sec:phasecut-review}
            This paper is inspired by a well-known approach to the phase retrieval problem known as \emph{PhaseCut}. The algorithm, developed in \cite{waldspurger_phase_2015}, casts the phase retrieval problem in terms of a semidefinite programming (SDP) relaxation of the maximum cut problem, which finds the maximum number of edges in a network that form a partition of the network into two disjoint subsets. The purpose of PhaseCut essentially boils down to the fact that the intractable modulus constraints in \eqref{eq:generalized-high-level-phir} can be avoided via a relaxed SDP. This is achieved by modeling the phase $\vu \in \C^m$ of the signal $|\mA \vx |$ as a rank-one matrix $\mX \triangleq \vu \vu^*$. With this matrix decision variable, \eqref{eq:generalized-high-level-phir} can be written as
            \begin{equation}
                \label{eq:semidefinite_relaxed_phase_retrieval}
                \minimize_{\mX \succeq 0 } \trace{ \mX \mM }, \quad \text{subject to:} \quad \rank(\mX) = 1, \quad \diag(\mX) = \mathds{1},
            \end{equation}
            which is relaxed by dropping the $\rank(\mX)=1$ constraint. Although we do not explicitly use this SDP relaxation in this work, its precursor quadratic formulation, \eqref{eq:basic-ph-retr-prob-fixed-phase} is central to the forthcoming results.

    \subsection{Review of Newton-Raphson power flow}
    \label{sec:nrpf-review}

    \subsubsection{Power flow equations}
        \label{sec:power_balance_equations}
            We consider an arbitrary electric power network with $n$ ``PQ" buses\textemdash i.e., buses where the complex power injections are specified. The well-known power flow equations in polar form for buses $i =1,\dots,n$ are
            \begin{subequations}
            \label{eq:power_balance}
                \begin{align}
                \label{eq:active_power_balance}
                    p_i &= v_i \sum_{k=1}^{n}  v_k (G_{ik} \cos\theta_{ik} + B_{ik}\sin\theta_{ik}) = p^g_i - p^d_i, \\
                \label{eq:reactive_power_balance}
                    q_i &= v_i \sum_{k=1}^{n}  v_k (G_{ik} \sin\theta_{ik} - B_{ik}\cos \theta_{ik}) = q^g_i - q^d_i,
                \end{align}
            \end{subequations}
            where $\theta_{ik}\triangleq\theta_i -\theta_k$ is the phase angle differences between buses $i,k$, $p^g_i,q^g_i$ are the generated active and reactive powers at bus $i$, $p^d_i,q^d_i$ are the demanded active and reactive powers  at bus $i$, $v_i,v_k$ are the voltage magnitudes at buses $i,k$, and $G_{ik},B_{ik}\in \mathbb{R}$ are the real and imaginary components of the $(i,k)$-th entry $Y_{ik} \triangleq G_{ik} + j B_{ik} \in \C$ of the nodal admittance matrix $\mY = \mG + j \mB \in \C^{n \times n}$. 

            By convention, the state of the system is defined as $\vx \triangleq [\vtheta^T,\vv^T]^T \in \R^{2n}$, where $\vv \in \R^n$ are the voltage magnitudes and $\vtheta \in (-\pi,\pi]^n$ are the voltage phase angles. The complex voltage at each bus takes the form $\overline{v}_i \triangleq v_i e^{j \theta_i} = v_i(\cos(\theta_i) + j \sin(\theta_i)) \in \C$. The bus active and reactive power injection vector $\vg(\cdot): \R^{2n}\mapsto \R^{2n}$, i.e., the vectorization of the power flow equations \eqref{eq:power_balance}, is defined as $\vg(\vx) \triangleq [\vp(\vx)^T,\vq(\vx)^T]^T \in \R^{2n}$. The solution $\vx^*$ of the non-linear system of equations \eqref{eq:power_balance} satisfies $\vg(\vx^*) - \left[(\vp^g - \vp^d )^T,(\vq^g - \vq^d)^T\right]^T = \boldsymbol{0}.$

    \subsubsection{Newton-Raphson power flow in polar coordinates}
    \label{sec:rectangular-review}
        The Newton-Raphson power flow algorithm is a well-studied iterative approach to solving the non-linear power flow equations \eqref{eq:power_balance}. This method consists of repeatedly solving the linear system of equations 
    \begin{equation}
    \renewcommand*{\arraystretch}{1.5}
    \label{eq:nr_pf_definition}
        \begin{bmatrix}
           \Delta \vp\\
           \Delta \vq
         \end{bmatrix} =
         \begin{bmatrix}
             \DPDTH(\vx) & \DPDV(\vx)\\
             \DQDTH(\vx) & \DQDV(\vx)
         \end{bmatrix}
            \begin{bmatrix}
            \Delta \vtheta\\
            \Delta \vv
        \end{bmatrix} = \mJ(\vx) \Delta \vx,
    \end{equation}
    where\textemdash \emph{around} a specified operating point\textemdash $\Delta \vp,\Delta \vq\in \R^n$ are small perturbations in the active and reactive power injections respectively, and $\Delta \vtheta,\Delta \vv \in \R^n$ are small perturbations in the voltage phase angles and magnitudes, respectively. The matrix $\mJ(\vx) \in \R^{2n \times 2n}$ in \eqref{eq:nr_pf_definition} is the \emph{power flow Jacobian}, which contains the partial derivatives of the power flow equations \eqref{eq:power_balance} with respect to the voltage magnitudes $\vv$ and phase angles $\vtheta$. 
    
    In this work, we consider an \emph{estimate} of the power flow Jacobian $\mJ(\vx_t)$ computed at a time point $t$. We use $\mJ(\vx_t)$ to provide an approximation \cite[(5.8)]{molzahn_survey_2019} of the AC power flow equations \eqref{eq:power_balance} around the operating point \emph{measured} at a time point $t$ as
    \begin{equation}
        \label{eq:first-order-linearization}
        \vg(\vx_{t+1}) \approx \vg(\vx_{t}) + \mJ(\vx_t)(\vx_{t+1} - \vx_{t}).
    \end{equation}
    Here, the power flow Jacobian $\mJ(\vx_t)$ denotes the \emph{sensitivity} of the injections to changes in the state from the operating point measured at time $t$. We consider the notion of observed or predicted perturbations in the injections and state around the operating point measured at time point $t$. Denote these perturbations as $\Delta \vg(\vx_t) \triangleq \vg(\vx_{t+1}) - \vg(\vx_t) = [\vp_{t+1}^T,\vq_{t+1}^T]^T - [\vp_{t}^T,\vq_{t}^T]^T$ and $\Delta \vx_t \triangleq \vx_{t+1} - \vx_t = [\vtheta_{t+1}^T,\vv_{t+1}^T]^T - [\vtheta_{t}^T,\vv_{t}^T ]^T$.

In this paper, the network model is \emph{unknown}, and PMU measurements as in \eqref{eq:pmu_signal_def} are \emph{entirely} unavailable. Thus, we can neither measure nor solve for the phase angles, and the quantities $\Delta \vtheta_t$, $\DPDTH(\vx_t)$, and $\DQDTH(\vx_t)$ are unknown. We will assume that the power-voltage magnitude sensitivities $\DPDV(\vx_t)$ and $\DQDV(\vx_t)$ are known in the sense that they can be estimated from prior measurements, which is a well-studied problem (e.g., see \cite{chen_measurement-based_2016,moffat_power_voltage_sensitivity}).

In this paper, we ask the following question: can we recover the phase angles $\vtheta_t$ as in \eqref{eq:pmu_signal_def} and the power flow Jacobian submatrices $\DPDTH(\vx_t), \DQDTH(\vx_t)$ as in \eqref{eq:nr_pf_definition}, using only the measurements $\vv_t,\vp_t,\vq_t$ as in \eqref{eq:ami_signal_definition}?  This question is illustrated in Fig. \ref{fig:nr-phret-scheme}. 
\begin{figure}[t]
    \centering
    \includegraphics[width=0.55\linewidth,keepaspectratio]{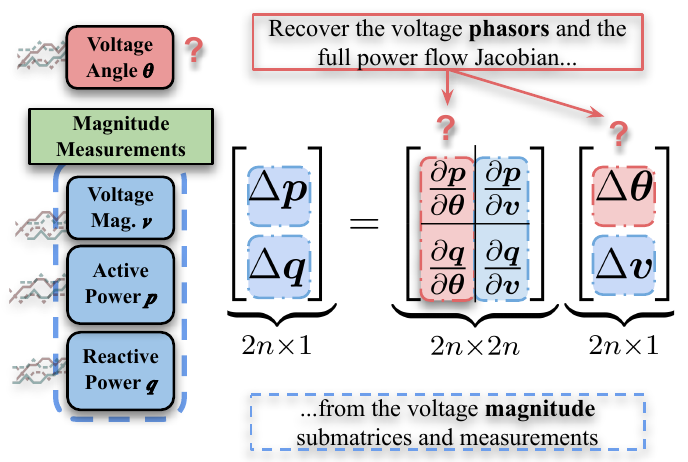}
    \caption{Illustration of voltage phase retrieval showing model-free power flow Jacobian reconstruction technique}
    \label{fig:nr-phret-scheme}
\end{figure}

    \subsection{The structure of the power flow Jacobian}
    \label{sec:jacobian_structure_review}

        \subsubsection{Power-voltage magnitude sensitivity matrices}
        \label{sec:power-vmag-sens-definition}
            Forming the entries of the block submatrices of the power flow Jacobian requires that we differentiate the power flow mismatch equations \eqref{eq:power_balance} with respect to the voltage magnitudes and phase angles at each bus. First, consider the derivatives of the net active power and reactive power injections at bus $i$ with respect to the voltage magnitudes at bus $k$:
            \begin{subequations}
            \label{eq:vmag_jac_blocks}
                \begin{align}
                \label{eq:dpdv_jac_block}
                    \frac{\partial p_i}{\partial v_k} = 
                    \begin{dcases}
                            v_i (G_{ik}\cos\theta_{ik} + B_{ik} \sin\theta_{ik} ) & i \neq k,\\
                            2 v_i G_{ii}  +  \sum_{k: k\neq i} v_k(G_{ik} \cos\theta_{ik} + B_{ik}\sin\theta_{ik})  & i = k.
                    \end{dcases}
                    \\
                \label{eq:dqdv_jac_block}
                    \frac{\partial q_i}{\partial v_k} = 
                    \begin{dcases}
                            v_i (G_{ik}\sin\theta_{ik} - B_{ik} \cos\theta_{ik} ) & i \neq k,\\
                            -2v_iB_{ii} + \sum_{k : k\neq i} v_k(G_{ik} \sin\theta_{ik} - B_{ik}\cos\theta_{ik}) & i = k.
                    \end{dcases}
                \end{align}
            \end{subequations}
            The diagonal terms can be alternatively expressed as $\frac{\partial p_i}{\partial v_i} = \frac{p_i(\vx)}{v_i}+ G_{ii} v_i $ and $\frac{\partial q_i}{\partial v_i} = \frac{q_i(\vx)}{v_i} - B_{ii}v_i$. 
        \subsubsection{Power-voltage phase angle sensitivity matrices}
            \label{sec:power-theta-sens-definition}
            Similarly to the power-voltage magnitude sensitivities in Section \ref{sec:power-vmag-sens-definition}, by differentiating the net active and reactive power injections with respect to the voltage phase angles, we obtain
            
            \begin{subequations}
            \label{eq:vph_jac_blocks}
                \begin{align}
                \label{eq:dpdth_jac_block}
                    \frac{\partial p_i}{\partial \theta_k} =
                    \begin{dcases}
                            v_i v_k (G_{ik} \sin\theta_{ik} - B_{ik}\cos\theta_{ik}) & i \neq k,\\
                            -v_i \sum_{k: k\neq i} v_k (G_{ik}\sin\theta_{ik} - B_{ik} \cos \theta_{ik}) & i = k,
                    \end{dcases}\\
            \label{eq:dqdth_jac_block}
                    \frac{\partial q_i}{\partial \theta_k} = 
                    \begin{dcases}
                        -v_i v_k(G_{ik}\cos \theta_{ik} + B_{ik}\sin\theta_{ik}) & i \neq k,\\
                        v_i \sum_{k: k\neq i} v_k (G_{ik} \cos\theta_{ik} + B_{ik} \sin \theta_{ik}) & i = k.
                    \end{dcases}
                \end{align}
            \end{subequations}
            The diagonal terms can alternatively be expressed as $\frac{\partial p_i}{\partial \theta_i}= -q_i(\vx) - B_{ii}v_i^2 $ and $\frac{\partial q_i}{\partial \theta_i} = p_i(\vx) - G_{ii}v_i^2 $.

\section{Analytical results}
\label{sec:analytical-results}
        
        \subsection{Model-free recovery of the power flow Jacobian}
        \label{sec:phaseless_symmetry}
            There are well-known mathematical symmetries that relate the power flow Jacobian block submatrices. These symmetries arise due to the physical laws of the power flow equations \cite[Ch. 9]{GRAINGER}. The canonical expressions  for these symmetries \cite[(9.5.4), (9.5.7)]{GRAINGER} depend on the network topology and electrical parameters. In this section, we propose a topology and parameter-free extension of these known symmetries that:
            \begin{enumerate}
                \item Has a valuable engineering implication, supporting system identification without voltage phase angles.
                \item Has a valuable theoretical implication, by providing structural constraints for applying the phase retrieval framework\textemdash informed by the physics of Newton-Raphson power flow.
            \end{enumerate}
            
            Specifically, we show that we can represent $\DPDTH$ and $\DQDTH$ as a function of the power-to-voltage magnitude submatrices of the power flow Jacobian $\DPDV$ and $\DQDV$, the voltage magnitudes $\boldsymbol{v}$, and the active and reactive power injections $\boldsymbol{p},\boldsymbol{q}$. Representing $\DPDTH$ and $\DQDTH$ in this way is valuable because:
            \begin{enumerate}
                \item It removes dependence on direct measurements of the voltage phase angles, as in \eqref{eq:pmu_signal_def}.
                \item It requires only measurements as in \eqref{eq:ami_signal_definition}, and \emph{estimates} for the power-voltage magnitude sensitivities; which is a well-studied problem \cite{chen_measurement-based_2016,moffat_power_voltage_sensitivity}.
            \end{enumerate}
            This result is achieved through the following Lemma:

            \begin{lemma}
            \label{lemma:phaseless-symmetry}
            Consider an electric power network with PQ buses $\{1,\dots,n\}$. Let $\vx = [\vtheta^T,\vv^T]^T$ be a guess for the network state, producing power injections $\vp(\vx),\vq(\vx)$. For all buses $i,k \in \{1,\dots,n\}$, the entries $\frac{\partial p_i}{\partial \theta_k}$ and $\frac{\partial q_i}{\partial \theta_k}$ of the power-voltage phase angle sensitivity matrices, $\DPDTH(\vx),\DQDTH(\vx) \in \mathbb{R}^{n \times n}$, as defined in \eqref{eq:vph_jac_blocks}, can be written as 
                \begin{subequations}
                \label{eq:thm1_expressions}
                    \begin{align}
                    \label{eq:thm1_dpdth}
                        \frac{\partial p_i}{\partial \theta_k} = \begin{dcases}
                        v_k \frac{\partial q_i}{\partial v_k} & i \neq k\\
                        v_i \frac{\partial q_i}{\partial v_i} - 2 q_i  & i = k,
                        \end{dcases}
                    \\
                    \label{eq:thm1_dqdth}
                        \frac{\partial q_i}{\partial \theta_k} = \begin{dcases}
                        -v_k \frac{\partial p_i}{\partial v_k} & i \neq k\\
                        -v_i \frac{\partial p_i}{\partial v_i} + 2p_i & i = k.
                        \end{dcases}
                    \end{align}
                \end{subequations}
            \end{lemma}
            \begin{proof}
            The relationship for the off-diagonal terms is well-known \cite{GRAINGER}, and is found by substituting \eqref{eq:dpdv_jac_block} and \eqref{eq:dqdv_jac_block} into \eqref{eq:dqdth_jac_block} and \eqref{eq:dpdth_jac_block}, respectively:
            \begin{align}
                \frac{\partial p_i}{\partial \theta_k} = 
                v_k\underbrace{
                    \left(
                        v_i(G_{ik} \sin \theta_{ik} - B_{ik} \cos \theta_{ik}
                    \right)
                }_{\triangleq \frac{\partial q_i}{\partial v_k}}, \quad i \neq k,\\
                \frac{\partial q_i}{\partial \theta_k} = 
                -v_k\underbrace{
                    \left( 
                        v_i(G_{ik} \cos \theta_{ik} + B_{ik} \sin \theta_{ik}) 
                    \right)
                }_{\triangleq \frac{\partial p_i}{\partial v_k}}, \quad  i\neq k.
            \end{align} 
            The topology-free expressions for the diagonal terms $\frac{\partial p_i}{\partial \theta_i}$ and $\frac{\partial q_i}{\partial \theta_i}$ in \eqref{eq:thm1_dpdth} and \eqref{eq:thm1_dqdth} are found by manipulating the power flow equations \eqref{eq:power_balance} and the definition of the power-voltage magnitude sensitivities \eqref{eq:vmag_jac_blocks}. To show this, first consider the expression \eqref{eq:thm1_dpdth}. We want to show that this holds for the case when $i=k$. Substituting the definition of the diagonal reactive power-voltage magnitude sensitivities $\frac{\partial q_i}{\partial v_i}$, \eqref{eq:dqdv_jac_block}, and the reactive power balance equation, \eqref{eq:reactive_power_balance}, suppose that we can write $\frac{\partial p_i}{\partial \theta_i}$ for $i=1,\dots,n$ as
            \begin{equation}
            \begin{split}
                  v_i\frac{\partial q_i}{\partial v_i}  - 2 q_i =
                -2v_i^2 B_{ii}  + v_i \sum_{k : k\neq i} v_k(G_{ik} \sin\theta_{ik} - B_{ik}\cos\theta_{ik}) 
                  - 2v_i \sum_{k=1}^{n}  v_k (G_{ik} \sin\theta_{ik} - B_{ik}\cos \theta_{ik}).
            \end{split}
            \end{equation}
            Next, pull out the $i$th term from the argument of the sum over $k=1$ to $n$, 
            \begin{equation}
               v_i(G_{ii}  \sin\theta_{ii} - B_{ii}\cos\theta_{ii}) = -v_iB_{ii},
            \end{equation}
            which allows us to express $\frac{\partial p_i}{\partial \theta_i}=   v_i \frac{\partial q_i}{\partial v_i} - 2q_i$ as
            \begin{equation}
            \label{eq:dpdth_proof_penultimate_step}
                \begin{split}
                   -2v_i^2 B_{ii}  &+ v_i 
                    \sum_{k: k \neq i} v_k(G_{ik} \sin\theta_{ik} - B_{ik}\cos\theta_{ik}) 
                     - 2v_i \sum_{k: k \neq i}  v_k (G_{ik} \sin\theta_{ik} - B_{ik}\cos \theta_{ik}) +2 v_i^2 B_{ii}.
                \end{split}
            \end{equation}
            Simplifying \eqref{eq:dpdth_proof_penultimate_step}, we obtain \eqref{eq:dpdth_jac_block} for $i=k$,
            \begin{align}
                   - v_i \sum_{k: k \neq i}  v_k (G_{ik} \sin\theta_{ik} - B_{ik}\cos \theta_{ik}) \triangleq \frac{\partial p_i}{\partial \theta_i},
            \end{align}
            as desired, which verifies that the expression \eqref{eq:thm1_dpdth} is true for $i=k$. 
            
            Analogously, we can see that the expression \eqref{eq:thm1_dqdth} for the reactive power-voltage phase angle sensitivity matrix is true by substituting in the expression \eqref{eq:dpdv_jac_block} for $\frac{\partial p_i}{\partial v_i}$, and the active power balance equation \eqref{eq:active_power_balance}. Then, applying the same logic as before, suppose that we can write $ \frac{\partial q_i}{\partial \theta_i}$ as:
            \begin{subequations}
            \label{eq:reactive_power_phase_sens_derivation}
                \begin{align}
                    \frac{\partial q_i}{\partial \theta_i}  = \ -v_i \frac{\partial p_i}{\partial v_i} &+ 2p_i,\\
                     \begin{split}
                          = -2 v_i^2 G_{ii}  &- v_i \sum_{k : k\neq i} v_k(G_{ik} \cos\theta_{ik} + B_{ik}\sin\theta_{ik})
                          +2v_i \sum_{k=1}^{n}  v_k (G_{ik} \cos\theta_{ik} + B_{ik}\sin\theta_{ik}),
                     \end{split}\\
                     \begin{split}
                          = -2 v_i^2 G_{ii}  &- v_i \sum_{k : k\neq i} v_k(G_{ik} \cos\theta_{ik} + B_{ik}\sin\theta_{ik})
                          +2v_i \sum_{k: k \neq i}  v_k (G_{ik} \cos\theta_{ik} + B_{ik}\sin\theta_{ik}) + 2v_i^2G_{ii}.
                     \end{split}
                \end{align}
            \end{subequations}
            Simplifying \eqref{eq:reactive_power_phase_sens_derivation} yields \eqref{eq:dqdth_jac_block} for $i=k$ as
            \begin{equation}
                v_i \sum_{k: k \neq i}  v_k (G_{ik} \cos\theta_{ik} + B_{ik}\sin\theta_{ik}) \triangleq \frac{\partial q_i}{\partial \theta_i},
            \end{equation}
            which verifies the desired expression \eqref{eq:thm1_dqdth} for $i=k$.
            \end{proof}

            Notice that by Lemma \ref{lemma:phaseless-symmetry} the power-voltage phase angle sensitivity matrices can be expressed as functions $\DPDTH,\DQDTH : \R^n \times \R^n \mapsto \R^{n \times n}$ of the form
            \begin{subequations}
            \label{eq:augmented-phase-sensitivities}
            \begin{empheq}[box=\tcbhighmath]{align}
                \DPDTH(\vv,\vq) = \diag(\vv) \DQDV - 2 \diag(\vq),\\
                \DQDTH(\vv,\vp) = -\diag(\vv)  \DPDV + 2 \diag(\vp),
            \end{empheq}
            \end{subequations}
            which are implicitly parameterized by $\DQDV$ and $\DPDV$.

    Furthermore, the matrices \eqref{eq:augmented-phase-sensitivities} are functions that depend solely on \eqref{eq:ami_signal_definition}, and depend on \emph{neither the network topology nor the voltage phase angles}. The matrices $\DPDV,\DQDV$ can also be efficiently estimated from \eqref{eq:ami_signal_definition} using well-studied methods \cite{chen_measurement-based_2016,moffat_power_voltage_sensitivity,chen_measurement-based_2014}.

    Computationally, we can verify expressions \eqref{eq:thm1_dpdth} and \eqref{eq:thm1_dqdth}  by considering a simple 2-bus network, shown in the ``flat start" configuration at the left of Fig. \ref{fig:nr-equivalence}. We solve the system \eqref{eq:nr_pf_definition} for the state $\vx_k = [\theta_2^{(k)}, v_2^{(k)}]^T$ iteratively for $k=1,\dots$ using the usual expressions \eqref{eq:vph_jac_blocks} and the phaseless expressions \eqref{eq:thm1_expressions} for $\frac{\partial p_2}{\partial \theta_2}$ and $\frac{\partial q_2}{\partial \theta_2}$. The equivalence of the mismatches yielded by this procedure is illustrated on the right of Fig. \ref{fig:nr-equivalence}. Further numerical evidence for the exactness of the expressions \eqref{eq:augmented-phase-sensitivities} in a large-scale test network is available in Appendix \ref{apdx:phaseless-nrpf-iter}.
    \begin{figure}[b]
        \centering
       \includestandalone[width=0.49\linewidth,keepaspectratio]{figs/circuit}\hfill 
       \includegraphics[width=0.48\linewidth,keepaspectratio]{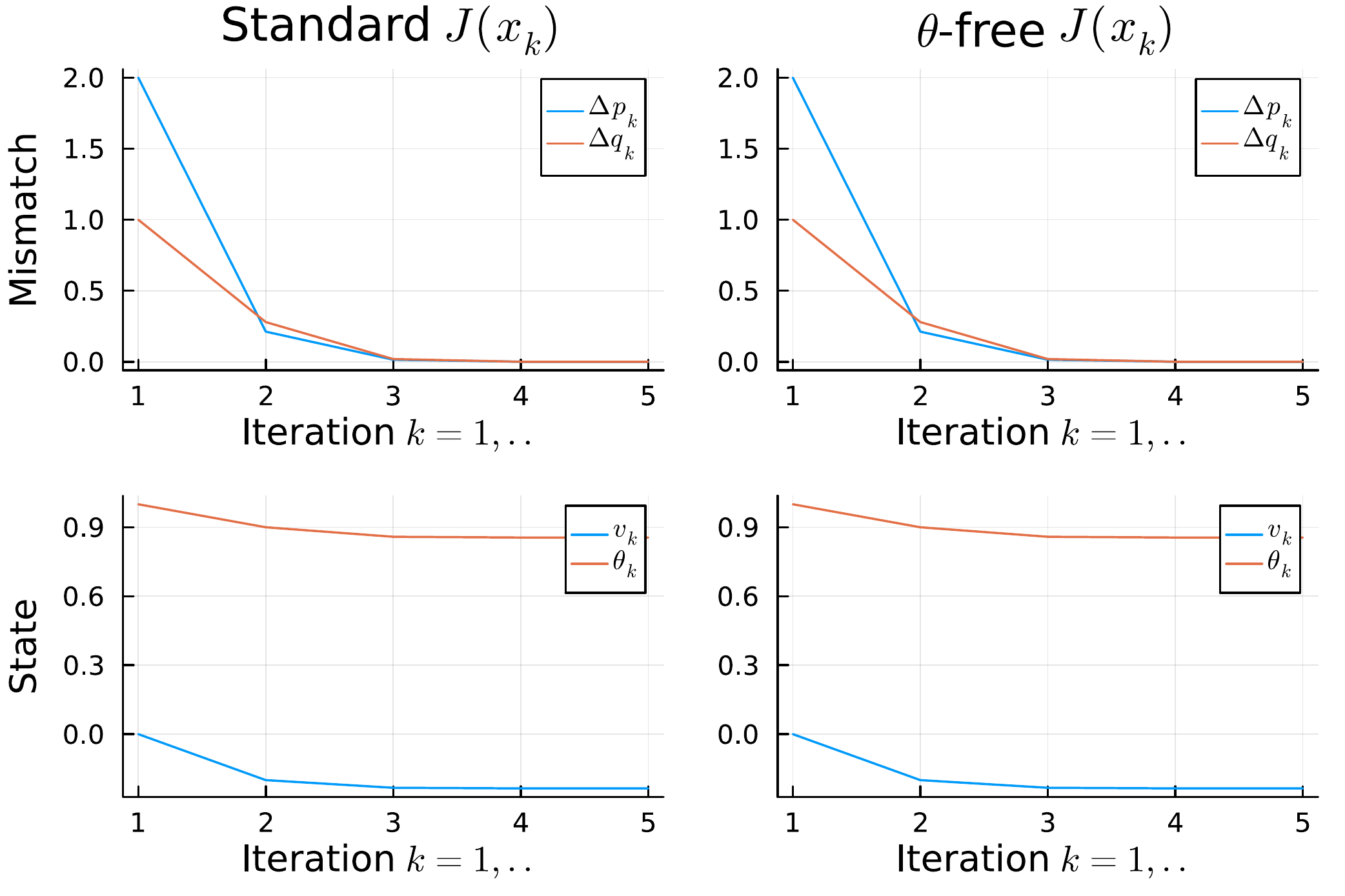}
        \caption{Equivalence of the Newton-Raphson iterations \eqref{eq:nr_pf_definition} using the $\theta$-free expressions \eqref{eq:thm1_expressions} and standard expressions \eqref{eq:vph_jac_blocks} for $\DPDTH,\DQDTH$ (right) for a simple two bus test case (left). }
        \label{fig:nr-equivalence}
    \end{figure}
    

    \subsection{Phase retrieval via power flow Jacobian recovery}
        \label{sec:phasecut-nrpf}
        We propose a specialized variant of the quadratic phase retrieval program \eqref{eq:basic-ph-retr-prob-fixed-phase} for electric power networks. The program leverages the structural symmetries of the power flow Jacobian described in Section \ref{sec:phaseless_symmetry}. 

        \subsubsection{Formulating the objective}
        Let $\vvbar_t \in \C^n$ be observed voltage phasors as in \eqref{eq:pmu_signal_def}, where $\overline{v}_{t}^{(i)} = v_{t}^{(i)} e^{j \theta_t^{(i)}} \triangleq v_t^{(i)} \angle \theta_t^{(i)} \in \C $ denotes the voltage phasor at bus $i \in \{1,\dots,n\}$. Suppose that we observe measurements $\vp_{t},\vq_{t} \in \R^n$ of active and reactive power injections, respectively; and of voltage magnitudes  $\vv_{t} \triangleq |\vvbar_{t}| \in \R^n$, as in \eqref{eq:ami_signal_definition} and Fig. \ref{fig:nr-phret-scheme}. We want to recover the phase angle perturbations $\Delta \vtheta_t \in (-\pi,\pi]^n$ using the observed perturbations $\Delta \vp_t$, $\Delta \vq_t$, and $\Delta \vv_t$. These are used in the linear approximation \eqref{eq:first-order-linearization}, taken around the network operating point. 
        
        Notice that we can equivalently represent a voltage phase angle $\theta_i$ as a complex number $u_i \in \C$. This is done by requiring that $u_i$ lies on the unit circle, and that its argument is the voltage phase angle $\theta_i$, i.e.,
        \begin{equation}
            |u_i|=1, \quad \textcolor{black}{ \theta_i = \operatorname{atan2}\left( \Im{u_i},\Re{u_i}\right) \quad i=1,\dots,n,} 
        \end{equation}
        where $\operatorname{atan2}(\Im{u_i},\Re{u_i})$ is the argument, or phase angle\footnote{
            The phase $\theta = \operatorname{atan2}(y,x)$ of $x + jy \in \C$ is defined as $\theta= \arctan(\frac{y}{x})\mathds{1}[x \neq 0] + (1-2\mathds{1}[y<0])(\pi \mathds{1}[x <0] + \frac{1}{2}\pi \mathds{1}[x =0])$, and is undefined if $x=y=0$. The operator $\mathds{1}[\cdot] =1$ if $[\cdot]$ is true, and $\mathds{1}[\cdot] =0$ otherwise.
        }, 
        of $u_i \in \C$.
     Drawing inspiration from the classical phase retrieval algorithm \eqref{eq:high-level-phir}, at each iteration $t$ we want to find the voltage phase angles $\vtheta_t$ that satisfy Euler's formula for a complex number on the unit circle $\vu_t$, i.e., $\vvbar_t = \diag(\vv_t) \vu_t(\vtheta_t),$ where
        \begin{equation}
             u_t^{(i)}\left(\theta_t^{(i)}\right) =  \cos\theta_t^{(i)} +j \sin\theta_t^{(i)}, \quad \left|u_t^{(i)}\right|=1, \ i=1,\dots,n.
        \end{equation}
        We will use this representation to directly apply the logic of the classical phase retrieval problem
        \begin{equation}
                \operatorname{find} \quad  \Delta \vtheta_t \in \R^n 
                \quad \text{subject to:} \quad \big| \diag(\Delta \vv_t) \vu_t(\Delta \vtheta_t) \big| = \Delta \vv_t,
        \end{equation}
        which consists of solving for a given iteration/sample $t$:
        \begin{equation}
        \label{eq:generic-nr-phase-recovery}
        \minimize_{\Delta \vtheta_t} \quad \norm{\Delta \vg(\vx_t) - \mJ(\vx_t)\Delta \vx_t}_2^2, \quad \text{subject to:} \quad \big| \diag(\Delta \vv_t) \vu_t(\Delta \vtheta_t) \big| = \Delta \vv_t.  
        \end{equation} 
        where $\Delta \vg(\vx_t) \triangleq \big[ \Delta \vp_t^T,\Delta \vq_t^T\big]^T$ denotes the observed perturbations $\big[ \Delta \vp_t^T,\Delta \vq_t^T\big]^T = \big[\vp_{t+1}^T, \vq_{t+1}^T\big]^T - \big[\vp_{t}^T, \vq_{t}^T\big]^T$ and $\Delta \vx_t \triangleq \big[ \Delta \vtheta_t^T,\Delta \vv_t^T\big]^T$ denotes the grid state change $ \big[ \Delta \vtheta_t^T,\Delta \vv_t^T\big]^T = \big[\vtheta_{t+1}^T,\vv_{t+1}^T\big]^T - \big[\vtheta_{t}^T,\vv_{t}^T \big]^T$.


        \subsubsection{Phase retrieval by exploiting the structure of the power flow Jacobian}
        Importantly, the formulation \eqref{eq:generic-nr-phase-recovery} assumes knowledge of the power-phase angle sensitivities $\DPDTH(\vx_t),\DQDTH(\vx_t)$, which is \emph{impossible} without knowledge of the network model or measurements of the form \eqref{eq:pmu_signal_def}. To resolve this, we can appeal to 
        Lemma \ref{lemma:phaseless-symmetry}. This allows us to reconstruct the full power flow Jacobian without knowledge of the voltage phase angles by considering the matrix decision variables $\DPDTH(\vv_t,\vq_t),\DQDTH(\vv_t,\vp_t) \in \R^{n \times n}$ in the phase retrieval program as in \eqref{eq:augmented-phase-sensitivities}
        
        The structural symmetries of the power flow Jacobian \eqref{eq:augmented-phase-sensitivities} can be understood as variants of the structural constraints and the real-valued phase constraints described in \cite[Section 3.6.2]{waldspurger_phase_2015} and \cite[Section 3.6.4]{waldspurger_phase_2015}, respectively.  The phase retrieval program for samples $t=1,\dots,$ can then be written as
        \begin{equation}
        \label{eq:nr-phret-program-with-symmetry-constraints}
        \begin{split}
        \renewcommand*{\arraystretch}{1.5}
            \minimize_{\Delta \vtheta_t,
            \DPDTH,\DQDTH}   \quad
           \norm{
            \begin{bmatrix}
                \Delta \vp_t\\
                \Delta \vq_t
            \end{bmatrix}
            -
            \begin{bmatrix}
                \DPDTH(\vv_t,\vq_t) & \DPDV\\
                \DQDTH(\vv_t,\vp_t) & \DQDV
            \end{bmatrix}
            \begin{bmatrix}
                \Delta \vtheta_t\\
                \Delta \vv_t
            \end{bmatrix}
            }_2^2 \quad 
            \text{subject to:} \quad \eqref{eq:augmented-phase-sensitivities}, 
        \end{split}
        \end{equation}
        where, as $\Delta \vtheta_t \in (-\pi,\pi]^n$ by construction, we have dropped the modulus constraints. Modern solvers can solve the program \eqref{eq:nr-phret-program-with-symmetry-constraints}  efficiently. Note that if $\DPDTH(\vv_t,\vq_t)$ is known and invertible, there is a unique solution for the phase angles. Dropping the arguments for simplicity, we have
        \begin{equation}
        \label{eq:delta-theta-recovery-expression}
            \Delta \vtheta_t = \left(\DPDTH\right)^{-1}\left( \Delta \vp_t - \DPDV \Delta \vv_t \right),
        \end{equation}
       because $ \Delta \vp_t = \DPDTH \Delta \vtheta_t + \DPDV \Delta \vv_t$, so $\DPDTH \Delta \vtheta_t = - \DPDV \Delta \vv_t + \Delta \vp_t$, and we have $\Delta \vtheta_t = -\DPDTH^{-1} \DPDV \Delta \vv_t + \DPDTH^{-1} \Delta \vp_t$ as desired. Similarly, if $\DQDTH(\vv_t,\vp_t)$ is invertible, by applying analogous steps for $\Delta \vq_t = \DQDTH \Delta \vtheta_t + \DQDV \Delta \vv_t$, we have
        \begin{equation}
            \label{eq:delta-theta-reactive-recovery-expression}
            \Delta \vtheta_t = \left( \DQDTH\right)^{-1}\left(  \Delta \vq_t -\DQDV \Delta \vv_t \right).
        \end{equation}

The expressions in \eqref{eq:thm1_expressions} and the program \eqref{eq:nr-phret-program-with-symmetry-constraints} enable the phase angle submatrices of the power flow Jacobian\textemdash and thus, the phase \textcolor{black}{angle} behavior itself\textemdash to be recovered using just the data in \eqref{eq:ami_signal_definition}. Furthermore, we will demonstrate the effectiveness of the  phase retrieval program \eqref{eq:nr-phret-program-with-symmetry-constraints} using estimated power-voltage magnitude sensitivities $\DPDV$ and $\DQDV$ in the numerical results in Section \ref{sec:numerical:baseline-comparison-knowlegdge-of-topology}.

  \subsubsection{Recovering the current phase angles}
  \label{sec:nr-phret-recovering-current-angles}
        For any guess for the voltage phase $\hat{\vtheta} \in (-\pi,\pi]^n$, let $\vvbar \triangleq \vv \angle \hat{\vtheta} \in \C^n$ denote the estimated complex voltage phasors. The complex current injection $\vell \in \C^n$ is related to $\vvbar$ and the complex power injection  $ \vs \triangleq \vp + j \vq \in \C^n$ through Ohm's law as $\vs = \vvbar \circ \conj{\vell}$, where $(\cdot)\circ(\cdot)$ denotes the elementwise multiplication of two vectors, and where $\conj{(\cdot)}$ denotes the elementwise complex conjugate.

        Once the voltage phase angles are recovered, all other quantities in the above relationship are given by the measurements in \eqref{eq:ami_signal_definition}. Therefore, we can compute an estimate for the complex current injections as 
        $\vell = \conj{\diag\left(\vvbar\right)^{-1} \vs}$. While we do not focus our estimation results on the current injections, we verify that this formulation is physically valid in Appendix \ref{apdx:phaseless-nrpf-iter}.

\subsection{Recovery guarantees based on Gershgorin discs}
\label{sec:sufficient-cond-unique}

The Gershgorin Circle Theorem provides regions of the complex plane where the eigenvalues of a matrix are guaranteed to lie. It has been applied in voltage stability assessment for siting and control of distributed energy resources  \cite{swartz_gershgorin_2022}, in an impedance-based stability requirement \cite{Liu_GERSHGORIN}, in small signal stability of swing equations \cite{amin_xu_stability}, and many other electric power systems applications. 

The motivations of using the Gershgorin Circle Theorem  to provide phase retrieval guarantees are twofold. First, the Theorem's  structure is intuitively similar to the power flow Jacobian symmetry expression \eqref{eq:thm1_expressions}. Second, the power flow Jacobian submatrices are empirically diagonally dominant, as we show numerically in Fig. \ref{fig:sensitivity-results} of the subsequent section. This property is favorable for applying the Gershgorin Circle Theorem.

Below, in Section \ref{sec:sufficient-voltage-phase-retrieval}, we apply Gershgorin to yield  sufficient conditions for the guaranteed recovery of the voltage phase angles. In Section \ref{sec:sufficient-condition-jac-invertible}, we apply a block matrix variant of Gershgorin to yield a sufficient condition for invertibility of the power flow Jacobian. The sufficient conditions of both Theorems provide these guarantees \emph{without} requiring any knowledge of the voltage phase angles $\vtheta$.

\subsubsection{Sufficient condition for guaranteed voltage phase retrieval}
\label{sec:sufficient-voltage-phase-retrieval}
\begin{theorem}
    \label{thm:disc-delta-theta}
    Let $\setB \subset \{1,\dots,n\}$ be a subset of PQ buses under study. Let $\vx = [\vtheta^T, \vv^T]^T \in \R^{2|\setB|}$ be a guess for the network state, let $\vp(\vx),\vq(\vx) \in \R^{|\setB|}$ be the observed power injections, and let $\DPDV(\vx),\DQDV(\vx) \in \R^{|\setB| \times |\setB|}$ be the power-voltage magnitude sensitivity matrices. For any perturbations in \eqref{eq:ami_signal_definition}, if for all $i \in \setB$
    \begin{subequations}
    \label{eq:thm1:q-conditions}
        \begin{align}
            \label{eq:thm1_row_offdiag}
       |q_i| &> \frac{1}{2} v_i \left( \sum_{k \in \setB \setminus\{i\}} \left| \frac{\partial q_k}{\partial v_i}\right| - \left| \frac{\partial q_i}{\partial v_i}\right|\right)\\
       \label{eq:thm1_column_offdiag} \textsf{or} \quad 
        |q_i| &> \frac{1}{2}\left(\sum_{k \in \setB \setminus \{i\}} v_k \left| \frac{\partial q_i}{\partial v_k}\right| - v_i \left| \frac{\partial q_i}{\partial v_i} \right| \right), 
        \end{align}
    \end{subequations}
    then there exists unique voltage phase angle perturbations $\Delta \vtheta$ such that \eqref{eq:delta-theta-recovery-expression} holds. Analogously, if for all $i \in \setB$
    \begin{subequations}
    \label{eq:thm1:p-conditions}
        \begin{align}
             |p_i| &> \frac{1}{2}v_i \left( \sum_{k \in \setB \setminus\{i\}} \left| \frac{\partial p_k}{\partial v_i}\right| - \left| \frac{\partial p_i}{\partial v_i} \right|\right)\\
             \textsf{or} \quad |p_i| &> \frac{1}{2} \left(\sum_{k \in \setB \setminus\{i\}} v_k \left| \frac{\partial p_i}{\partial v_k}\right| - v_i \left| \frac{\partial p_i}{\partial v_i}\right| \right),
        \end{align}
    \end{subequations}
    then there exists unique voltage phase angle perturbations $\Delta \vtheta$ such that \eqref{eq:delta-theta-reactive-recovery-expression} holds.
\end{theorem}
The proof of Theorem \ref{thm:disc-delta-theta} can be found in Appendix \ref{apdx:thm1-proof}.

Theorem \ref{thm:disc-delta-theta} is valuable because it depends solely on measurements of the form \eqref{eq:ami_signal_definition}, and the matrices $\DPDV,\DQDV$, which can be efficiently estimated from \eqref{eq:ami_signal_definition} using well-studied methods. Theorem \ref{thm:disc-delta-theta} gives explicit bounds on the magnitude of the active or reactive power injection at each bus in a chosen set of buses $\setB$, such that the phase angles at these buses can be guaranteed to be uniquely recovered. 

Furthermore, the bounds are each functions of only the voltage magnitudes and one of \emph{either} the active or reactive power injections, respectively.  If the conditions in Theorem \ref{thm:disc-delta-theta} hold, then $\Delta \vtheta$ can be uniquely recovered via either the expression in \eqref{eq:delta-theta-recovery-expression} or \eqref{eq:delta-theta-reactive-recovery-expression}, which are equivalent by construction.

\subsubsection{Sufficient condition for power flow Jacobian invertibility and guaranteed voltage phase retrieval}
\label{sec:sufficient-condition-jac-invertible}
We can establish conditions for the invertibility of the \emph{full} power flow Jacobian without knowledge of the phase angles by appealing to an analogue of the Gershgorin Circle Theorem for block matrices \cite{feingold_block_1962,echeverria_block_2018}. To yield this result, we introduce the following assumption:
\begin{assumption}
\label{assum:full-rank-diagonal-blocks}
The matrices $\DPDTH(\vv,\vq)$ and $\DQDV$ are full rank. 
\end{assumption}

Assumption \ref{assum:full-rank-diagonal-blocks} is not restrictive, as it only requires that the diagonal blocks of the Jacobian are invertible, without certification of the invertibility of the entire Jacobian. Moreover, we show in Table \ref{tab:disc-thm2} of Appendix \ref{apdx:thm2-proof} that the assumption holds for all the test cases we study, including very large-scale, realistic test cases. With Assumption \ref{assum:full-rank-diagonal-blocks} in hand, we can use the matrices \eqref{eq:augmented-phase-sensitivities} to yield the following condition, which guarantees the non-singularity of the full Jacobian given the measurements in \eqref{eq:ami_signal_definition}.

\begin{theorem}
\label{thm:disc}
Let $\setB \subset \{1,\dots,n\}$ be a subset of PQ buses under study with a guess for the network state $\vx$ and injections $\vp(\vx),\vq(\vx)$ as in Theorem \ref{thm:disc-delta-theta}. Consider the voltage phase angle sensitivity matrix functions $\DPDTH(\vv,\vq)$ and $\DQDTH(\vv,\vp)$ as defined in \eqref{eq:augmented-phase-sensitivities}. Suppose that Assumption \ref{assum:full-rank-diagonal-blocks} holds. Then, if
\begin{subequations}
    \begin{align}
    \label{eq:thm2_row_offdiag}
        \norm{\left(\DPDTH(\vv,\vq)\right)^{-1} \DPDV}_2 < 1 \quad \textsf{and} \quad \norm{\left(\DQDV\right)^{-1} \DQDTH(\vv,\vp)}_2 <1,\\
     \label{eq:thm2_column_offdiag}
    \textsf{or,} \quad  \norm{\left(\DPDTH(\vv,\vq)\right)^{-1} \DQDTH(\vv,\vp)}_2 <1  \quad \textsf{and} \quad \norm{\left(\DQDV\right)^{-1} \DPDV}_2 < 1,
    \end{align}
\end{subequations}
then the power flow Jacobian $\mJ(\vx) \in \R^{2|\setB| \times 2|\setB|}$ is non-singular, i.e., $\rank\left(\mJ(\vx)\right) = 2|\setB|$. Thus, the network is not in a state of voltage collapse \textcolor{black}{due to Jacobian singularity, (e.g., saddle node bifurcation),} and there exists unique voltage phase angle perturbations $\Delta \vtheta$ such that $[\Delta\vtheta^T, \Delta \vv^T]^T = \mJ(\vx)^{-1}[\Delta \vp^T,\Delta \vq^T]^T$ for any perturbations in \eqref{eq:ami_signal_definition}.
\end{theorem}
The proof of Theorem \ref{thm:disc} can be found in Appendix \ref{apdx:thm2-proof}.

Theorem \ref{thm:disc} is a stronger condition than that of Theorem \ref{thm:disc-delta-theta}. If Theorem \ref{thm:disc} holds, then $\mJ(\vx)$ is guaranteed to be invertible. This is also a sufficient condition for the network to \emph{not} be in a state of voltage collapse due to power flow Jacobian singularity, see \cite{hiskens_singularity,grijalva_singularity,grijalva_individual_2012} for more details. However, note that Theorem \ref{thm:disc} requires all the measurements in \eqref{eq:ami_signal_definition} for each condition. In contrast, each condition of Theorem \ref{thm:disc-delta-theta} only requires the voltage magnitudes and one of either the reactive power or active power injections.

\section{Numerical results}

    In this section, we numerically validate the proposed algorithms developed in this paper. A Julia package based on \texttt{PowerModels.jl} \cite{powermodels} has been developed for this work, and is publicly available.
    \begin{center}
    \url{https://github.com/samtalki/PowerPhaseRetrieval.jl}
    \end{center}
    We first present numerical analyses of the recovery guarantees in Section \ref{sec:numres:recovery-guarantees}.
    Then, we present numerical experiments for the phase retrieval program \eqref{eq:nr-phret-program-with-symmetry-constraints}. In Section \ref{sec:phret-results}, we  perform phase retrieval for a single time point at a nominal operating point specified by the network's AC power flow solution, with noise added. For this reason, the subscript $t$ is often dropped. In Section \ref{sec:impact-measurement-nonidealities}, actual time-series measurement data are used to apply the algorithm sequentially, including nonidealities such as asynchronous measurements, delays, and noise.

    \subsection{Validation of recovery guarantees}
    \label{sec:numres:recovery-guarantees}
    In addition to the proofs of Theorem \ref{thm:disc-delta-theta} and Theorem \ref{thm:disc} that are available in  Appendix \ref{apdx:thm1-proof} and Appendix \ref{apdx:thm2-proof}, respectively, we report numerical validation of these results for several \textsc{Matpower} \cite{zimmerman_matpower} test networks. \textcolor{black}{This includes the widely used \texttt{RTS\_GMLC} \cite{barrows_ieee_2019} test network, and the large-scale, realistic Texas A\&M synthetic grids \cite{birchfield_grid_2017}}. The quantities used for validating the theorem were computed with \texttt{PowerModels.jl} \cite{powermodels}. 
    
    In Table \ref{tab:disc-thm1}, we show the percentage of the PQ buses that are in the largest subset of PQ buses $\setB$ that satisfy Theorem \ref{thm:disc-delta-theta}. For the buses that do not satisfy Theorem \ref{thm:disc-delta-theta}, we show the maximum length $r_{\sf worst}>0$ of the radii segment of the discs that crosses into the left-hand side of the complex plane.

    In Table \ref{tab:disc-thm2}, we show the percentage of the PQ buses that are in the largest subset of PQ buses $\setB$ that satisfy Theorem \ref{thm:disc} (see the end of the proof of Theorem \ref{thm:disc} in Appendix \ref{apdx:thm2-proof} for more details). For each test network, regardless of whether Theorem \ref{thm:disc} is satisfied for the entire set of PQ buses, we show the maximum spectral norm of the four matrix arguments in \eqref{eq:thm2_row_offdiag}, \eqref{eq:thm2_column_offdiag}.

\begin{table}[th]
    \centering
    \begin{tabular}{|c|c|c|c|}
    \hline
         Case & \# PQ Buses & \% Satisfying Thm. \ref{thm:disc-delta-theta} & $r_{\sf worst}$  \\
         \hline
         \texttt{14} & 9 & 100.0\% & \textemdash\\
         \texttt{24\_ieee\_rts} & 13 & 100.0\% & \textemdash\\
         \texttt{ieee30} & 24 & 95.83\% & $1.4 \times 10^{-14}$\\
         \texttt{RTS\_GMLC} & 40 & 100.0\% & \textemdash\\
         \texttt{118} & 64 & 100.0\% & \textemdash\\
        \texttt{89pegase} & 77 & 94.81\% & 8.32\\ 
         \texttt{ACTIVSg200} & 162 & 96.91\% & 0.088\\
         \texttt{ACTIVSg500} & 444 & 94.37\% & 3.014\\
         \texttt{ACTIVSg2000} & 1608 & 84.83\% & 28.31\\
        \hline
    \end{tabular}
    \caption{Analysis of Theorem \ref{thm:disc-delta-theta} for PQ buses of various test cases}
    \label{tab:disc-thm1}
\end{table}

\begin{table}[th]
    \centering
    \begin{tabular}{|c|c|c|c|}
    \hline
         Case & \# PQ Buses & \% Satisfying Thm. \ref{thm:disc} & $\sigma_{\sf max}$  \\
         \hline
         \texttt{14} & 9 & 100.0\% & 0.876\\
         \texttt{24\_ieee\_rts} & 13 & 100.0\% & 0.401 \\
         \texttt{ieee30} & 24 & 95.83\% & 1.437 \\
         \texttt{RTS\_GMLC} & 40 & 100.0\% & 0.444\\
         \texttt{118} & 64 & 100.0\% & 0.473 \\
        \texttt{89pegase} & 77 & 100\% & 0.954\\
         \texttt{ACTIVSg200} & 162 & 100\% & 0.698 \\
         \texttt{ACTIVSg500} & 444 & 99.77\% & 1.090\\
         \texttt{ACTIVSg2000} & 1608 & 99.69\% & 1.180 \\
        \hline
    \end{tabular}
    \caption{Analysis of Theorem \ref{thm:disc} for PQ buses of various test cases}
    \label{tab:disc-thm2}
\end{table}

    \subsection{Recovery of voltage phase angles and power-voltage phase angle sensitivities} 
    \label{sec:phret-results}

    \subsubsection{Simulated measurements}
    \label{sec:simultaned-measurements}
    To test the phase retrieval program \eqref{eq:nr-phret-program-with-symmetry-constraints} on several standard \textsc{Matpower} \cite{zimmerman_matpower} test cases, we simulate noisy measurements. We used \texttt{PowerModels.jl} \cite{powermodels} to solve the AC power flow equations for each test case, yielding a nominal solution for the grid state $\vx^* = [\vtheta^T,\vv^T]^T$, with corresponding power flow Jacobian $\mJ(\vx^*)$. To provide a simulation of the measurements in \eqref{eq:ami_signal_definition}, we computed observed injection perturbations as
    \begin{equation}
    \label{eq:noisy-mismatches}
        \Delta \vg^{\sf obs}(\vx^*) \triangleq \mJ(\vx^*) \vx^* + \begin{bmatrix} \vxi^p \\ \vxi^p \end{bmatrix},
    \end{equation}
    where the entries of the noise vectors $\vxi^p,\vxi^q  \in \R^n$ are distributed as $\xi^p_i,\xi_i^q \sim \setN(0,\sigma_{\sf meas}^2)$ for all $i \in \{1,\dots,n\}$. We then computed the state perturbations as $\Delta \vx^{\sf obs} \triangleq \mJ(\vx^*)^{-1}\Delta \vg^{\sf obs}$.

    \subsubsection{Recovery of voltage phase angles}
    \label{sec:numerical-ph-angle-recovery}
    We solved the phase retrieval program \eqref{eq:nr-phret-program-with-symmetry-constraints} for various \textsc{Matpower} \cite{zimmerman_matpower} test networks: the IEEE RTS 24-bus test case, the IEEE 30-bus test case, and the Grid Modernization Laboratory Consortium Reliability Test System (\texttt{RTS\_GMLC}) case \cite{barrows_ieee_2019}. In Figure \ref{fig:phret_results_by_noise}, we show the estimated PQ bus voltage phase angles at the AC power flow solution of \texttt{RTS\_GMLC} for multiple noise levels, as in \eqref{eq:noisy-mismatches}. The ground truth bus voltage phase angles are shown in black. The results verify that the phase retrieval program \eqref{eq:nr-phret-program-with-symmetry-constraints} is effective and robust to noise. Additional experiments for other test cases are available in Appendix \ref{apdx:additional-numres}.
    \begin{figure}[htb]
        \centering
        \includegraphics[width=0.55\linewidth,keepaspectratio]{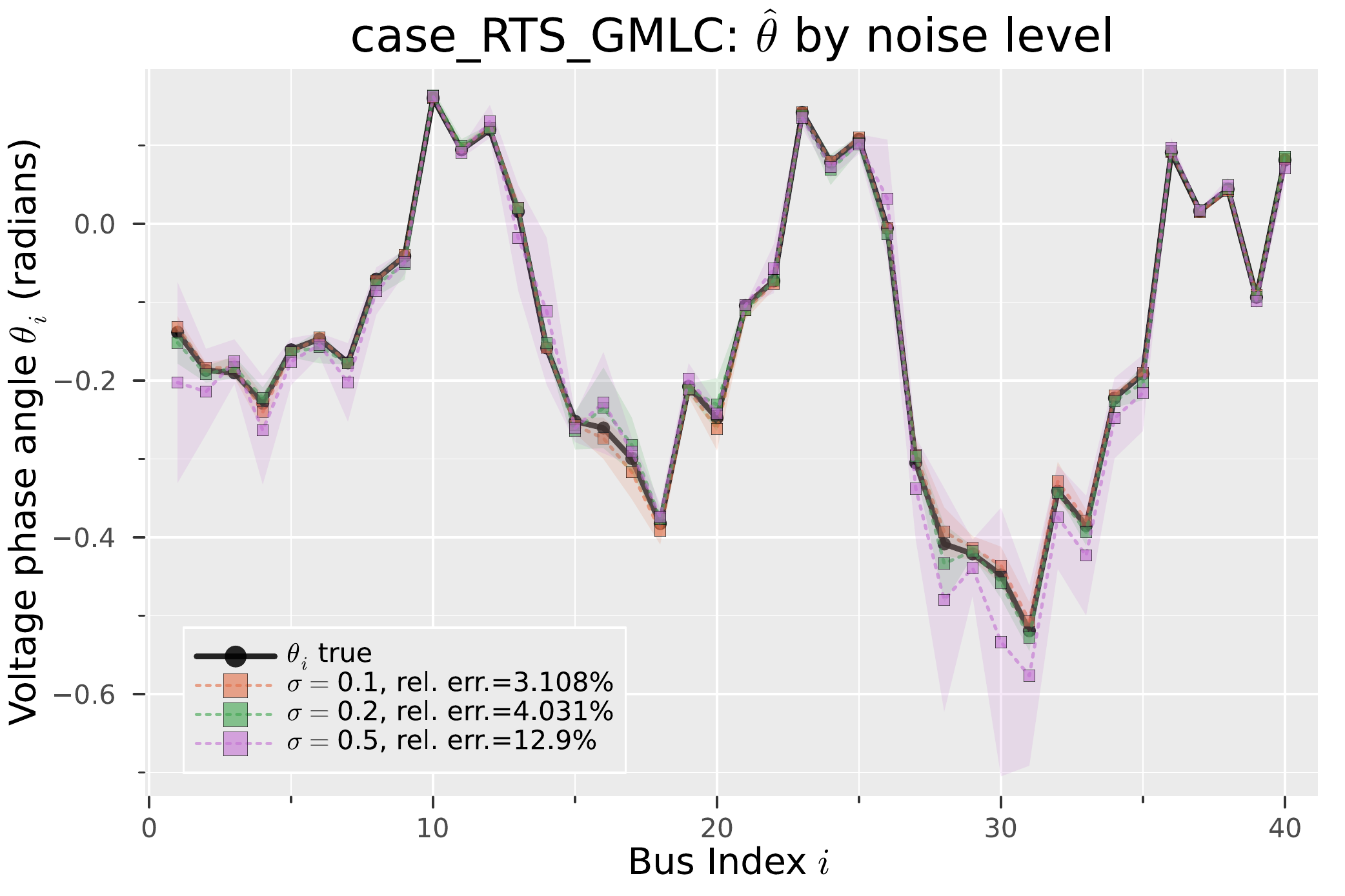}
        \caption{Performance of \eqref{eq:nr-phret-program-with-symmetry-constraints} by measurement noise level. The shaded ribbon indicates a $\pm 2 e$ error interval where $e$ is the absolute bus voltage phase error.}
        \label{fig:phret_results_by_noise}
    \end{figure}

\subsubsection{Recovery of the power-voltage phase angle sensitivity matrices}
    \label{sec:sensitivity-results}
        The structural symmetries of the power flow Jacobian \eqref{eq:thm1_expressions}  that were derived in Section \ref{sec:jacobian_structure_review} allow us to recover the power-voltage phase angle sensitivity matrices, i.e., the two unknown blocks of the power flow Jacobian. This is achieved by simultaneously recovering the phase angle perturbations and the Jacobian submatrix decision variables $\DPDTH,\DQDTH$ in the phase retrieval program \eqref{eq:nr-phret-program-with-symmetry-constraints}, subject to the structural constraints \eqref{eq:thm1_expressions}. In Fig. \ref{fig:sensitivity-results}, we show the optimal values of the power-voltage phase angle sensitivity matrices obtained by solving the phase retrieval program \eqref{eq:nr-phret-program-with-symmetry-constraints}, compared with the deterministic matrices defined in \eqref{eq:vph_jac_blocks}.
        \begin{figure}[t!]
            \centering
            \includegraphics[width=0.55\linewidth,keepaspectratio]{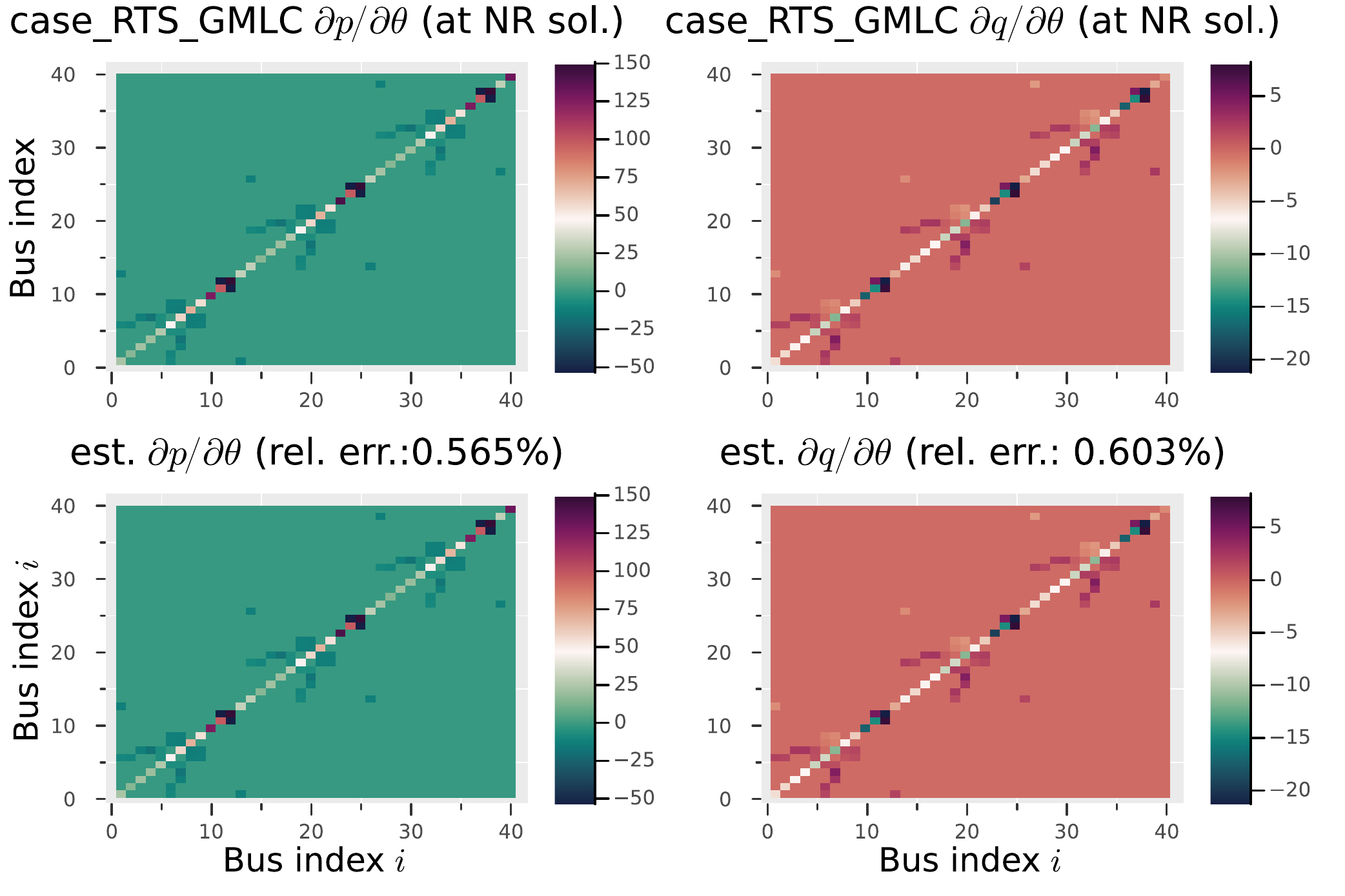}
            \caption{Recovery of the power-phase angle submatrices $\DPDTH,\DQDTH$ of the power flow Jacobian for the \texttt{RTS\_GMLC} network via the phase retrieval program \eqref{eq:nr-phret-program-with-symmetry-constraints} with $\sigma_{\sf meas} = 0.1$. Bus indices correspond to PQ bus ordering of the basic network format of \texttt{PowerModels.jl} \cite{powermodels}. Frobenius norm relative error is displayed on the second row.}
            \label{fig:sensitivity-results}
        \end{figure}

    \subsubsection{Robustness to errors in the estimates of the power-voltage magnitude sensitivity matrices}
    \label{sec:phret-results-robust}
        We introduce errors to the power-voltage magnitude sensitivity matrices to simulate the application of the phase retrieval program \eqref{eq:nr-phret-program-with-symmetry-constraints} using sensitivity matrices $\DPDV,\DQDV$ that were pre-estimated from past measurements of \eqref{eq:ami_signal_definition}. These errors are simulated by perturbing the deterministic, model-derived matrices with an additive Gaussian random matrix as
        \begin{equation}
        \label{eq:simulated-jacobian-estimation-error}
            \DPDV^{\sf obs} = \DPDV(\vx^*) + \mXi^{\sf pv}, \quad 
            \textsf{and} \quad \DQDV^{\sf obs} = \DQDV(\vx^*) + \mXi^{\sf qv},
        \end{equation}
        where $\Xi^{\sf pv}_{i,k},\Xi^{\sf qv}_{i,k} \sim \setN(0,\sigma^2_{\sf jac})$ for all $i,k \in \{1,\dots,n\}$. Matrices $\DPDV^{\sf obs}$ and $\DQDV^{\sf obs}$ are used as inputs in the phase retrieval program \eqref{eq:nr-phret-program-with-symmetry-constraints} and the power-phase angle sensitivity functions \eqref{eq:augmented-phase-sensitivities} in Section \ref{sec:impact-measurement-nonidealities}.

    \subsection{Impact of measurement nonidealities}
    \label{sec:impact-measurement-nonidealities}
        In practice, several nonidealities may be present in measurements of the form \eqref{eq:ami_signal_definition}, such as:
        \begin{enumerate}
            \item The measurements may not be instantaneous. They may be averaged measurements over a particular sampling period, such as one hour. 
            \item  The measurements may not be time-synchronized. They may possess delays and asynchronous sampling.
        \end{enumerate}
        Below, we assess the impact of these nonidealities on the proposed phase retrieval scheme. In Section \ref{sec:impact-time-synch}, we subsample measurements provided by \cite{barrows_ieee_2019}, and also propose a model to analyze the impact of random delays and asynchronous measurements on the method. Simulations are then provided in Section \ref{sec:nonideal-meas-simulations}.
        
 \begin{figure*}[!t]
            \centering
            \includegraphics[width=0.49\linewidth,keepaspectratio]{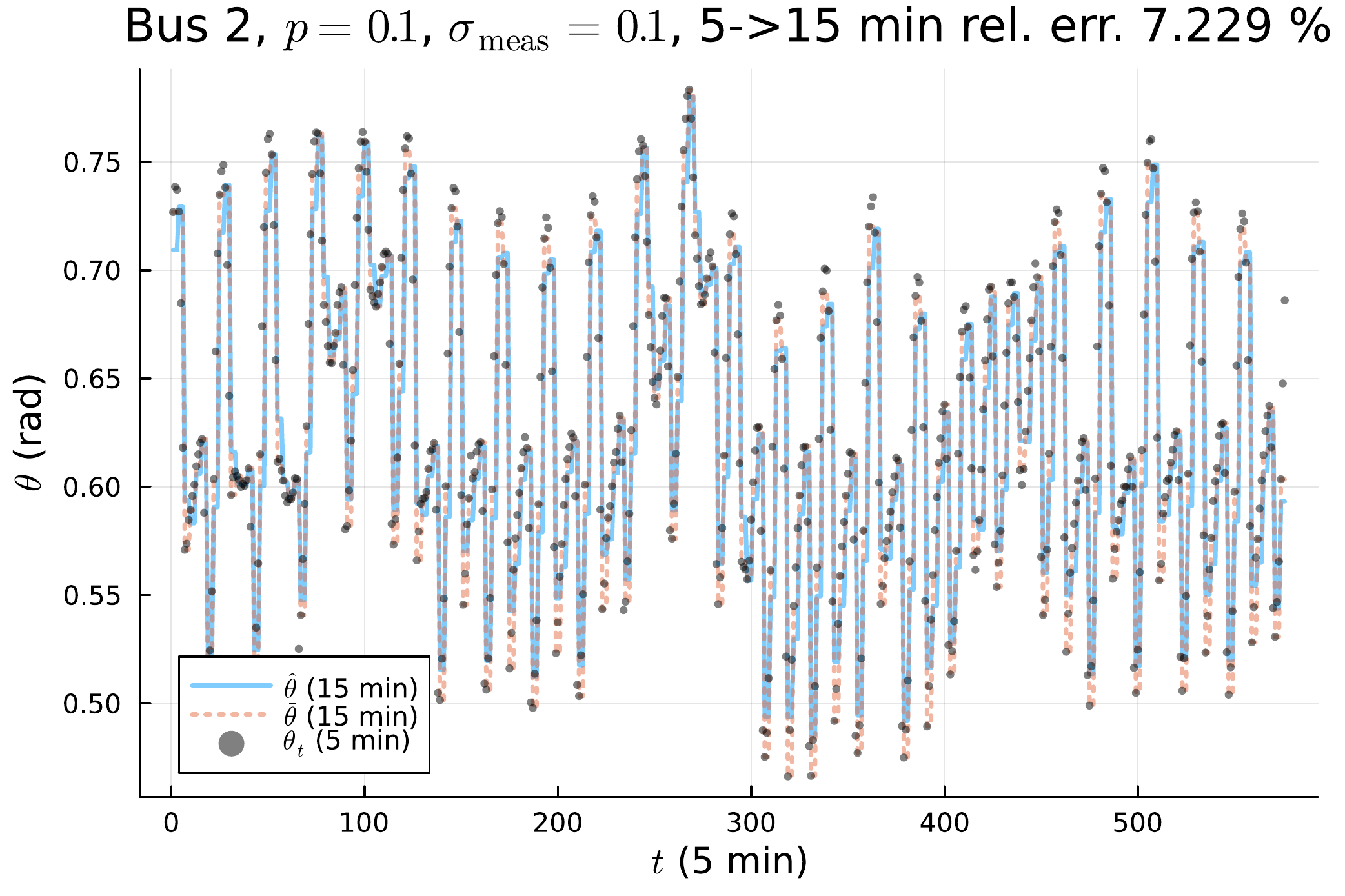}
            \includegraphics[width=0.49\linewidth,keepaspectratio]{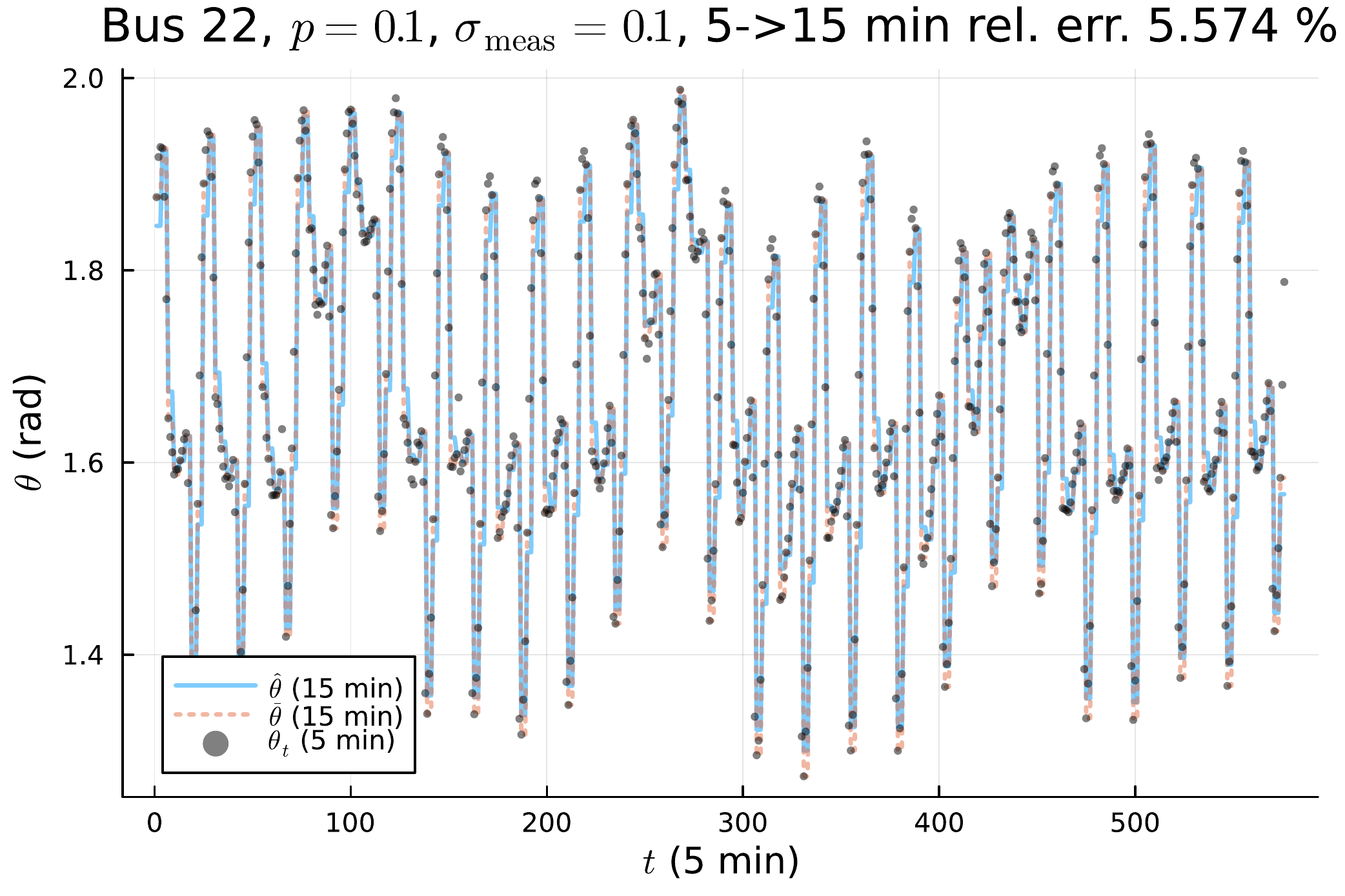}
            \caption{\textcolor{black}{Ground truth (blue) and estimated (orange dashed) voltage phase angles at 15 min. granularity, juxtaposed with ground truth 5 min. granularity voltage phase angles (black dots). Shown are PQ bus 2 ``AGRICOLA" (left) and PQ bus 22 ``BARKLA" (right) of the \texttt{RTS\_GMLC} network model.}}
            \label{fig:granularity-delayed}
        \end{figure*}
        \subsubsection{Modeling asynchronous measurements and non-ideal sampling rates}
        \label{sec:impact-time-synch}
        In practice, measurements of the form \eqref{eq:ami_signal_definition} regularly have delays and may be asynchronous. A simple model of this phenomenon can be obtained via a Bernoulli random vector $\vd \in \{0,1\}^n$, where the entries are i.i.d. as $d_i \sim \textsf{Bern}(p) $ for nodes $i=1,\dots,n$. In this model, the measurement data stream of node $i$ is delayed  with probability $p$, and is not delayed with probability $1-p$.

        Results of a survey of 31 distribution utilities in 2021 found that $>$60\% of respondents have access to AMI at sampling rates at least as frequent as 15 minutes \cite[Table 2]{peppanen_distribution_2021}, while a smaller percentage had 5 minute sampling rates. We model this nonideality by subsampling 5-minute ground truth measurements to 15-minute granularity.

        \subsubsection{Simulations}
        \label{sec:nonideal-meas-simulations}
        In Fig. \ref{fig:granularity-delayed}, the impacts of time delays, data rates, and measurement error are depicted. The ground-truth 5-minute time-series of active and reactive power load measurements are those provided by the \texttt{RTS\_GMLC} network model \cite{barrows_ieee_2019}. These ground truth measurements were used to sequentially solve the AC power flow equations, and obtain the voltage magnitudes and phase angles at 5-min intervals over the course of a year.  We then recover the voltage phase angles by applying the program \eqref{eq:nr-phret-program-with-symmetry-constraints} sequentially using measurements of the form \eqref{eq:ami_signal_definition}. The measurements were corrupted in the following ways:
        \begin{enumerate}
            \item The 5-minute ground truth $\vv_t,\vp_t,\vq_t$ measurements are subsampled to 15 min. granularity measurements.
            \item Node $i$ is delayed by 15 minutes relative to the true time-series with probability $p=0.1$, modeled by the Bernoulli random variable $\vd$. If $d_i=1$, then $\texttt{circshift}(\cdot)$ is applied to each measurement stream of bus $i$.
            \item Additive white Gaussian noise is added as in Section \ref{sec:simultaned-measurements}, with $\sigma_{\sf meas} = 0.1$.
        \end{enumerate}
        The 5-min granularity ``ground-truth" phase angles at the AC power flow solutions provided by \texttt{PowerModels.jl} are shown as black dots. The title depicts the relative percentage error between the recovered phase angle time series (interpolated back to 5-min granularity) and the 5-min ground truth.

        \subsection{Baseline comparison with known topology}
        \label{sec:numerical:baseline-comparison-knowlegdge-of-topology}

        We compare the method with a standard state estimation approach where the topology is known. We achieve this comparison by using exact closed form expressions \cite[5.2.1]{molzahn_survey_2019} for the complex-power to voltage magnitude and complex power-voltage phase angle sensitivities that depend on knowledge of the topology $\mY$. These matrices $\DSDV,\DSDTH \in \C^{n \times n}$ are defined as 
        \begin{align}
            \label{eq:dsdv-def} \DSDV &\triangleq \diag(\vvbar)\left( \diag(\conj{\mY} \conj{\vvbar}) + \conj{\mY} \diag(\conj{\vvbar})\right) \diag(\vv)^{-1},\\
            \label{eq:dsdth-def} \DSDTH &\triangleq j \diag(\vvbar)\left( \diag(\conj{\mY} \conj{\vvbar}) - \conj{\mY} \diag(\conj{\vvbar})\right),
        \end{align}
        where $\conj{(\cdot)}$ denotes the complex conjugate. Using \eqref{eq:dsdv-def} and \eqref{eq:dsdth-def}, closed form expressions for the power flow Jacobian submatrices are given as
        $\DPDV= \Re{\DSDV}$, $\DQDV =\Im{\DSDV}$,  $\DPDTH=\Re{\DSDTH}$, and $\DQDTH = \Im{\DSDTH}$.

         The above matrices can then be used to build a deterministic (model-based) Jacobian $\mJ^{\sf model} \in \R^{2n \times 2n}$ to directly solve a \textcolor{black}{standard} state estimation problem of the form
       \begin{equation}
           \minimize_{\Delta \vtheta} \quad \norm{\mJ^{\sf model} [\Delta \vtheta^T, \Delta \vv^T]^T - \Delta \vg^{\sf obs} }_2^2,
       \end{equation}
       where $\DPDTH,\DQDTH$ are known and are not decision variables, and $\Delta \vv \triangleq \Delta \vv^{\sf obs}$. \textcolor{black}{This stands in contrast with the program \eqref{eq:nr-phret-program-with-symmetry-constraints}, where $\DPDTH,\DQDTH$ are \emph{unknown}} and are included as decision variables.

        In Fig. \ref{fig:known-topology-comparison}, we benchmark the performance of the phase retrieval program \eqref{eq:nr-phret-program-with-symmetry-constraints} that does not require a model against this standard state estimation approach with a known model. The quantities $\sigma_{\sf meas}$ and $\sigma_{\sf jac}$ refer to the standard deviation of the additive Gaussian noise applied to $\Delta \vg^{\sf obs}$ and $\DPDV^{\sf obs},\DQDV^{\sf obs}$, respectively, as denoted in \eqref{eq:noisy-mismatches} and \eqref{eq:simulated-jacobian-estimation-error}, respectively. Note that in Fig. \ref{fig:known-topology-comparison}, we are applying noise to \emph{both} the measurements and $\DPDV,\DQDV$.

        \begin{figure}[tb]
            \centering
            \includegraphics[width=0.6\linewidth,keepaspectratio]{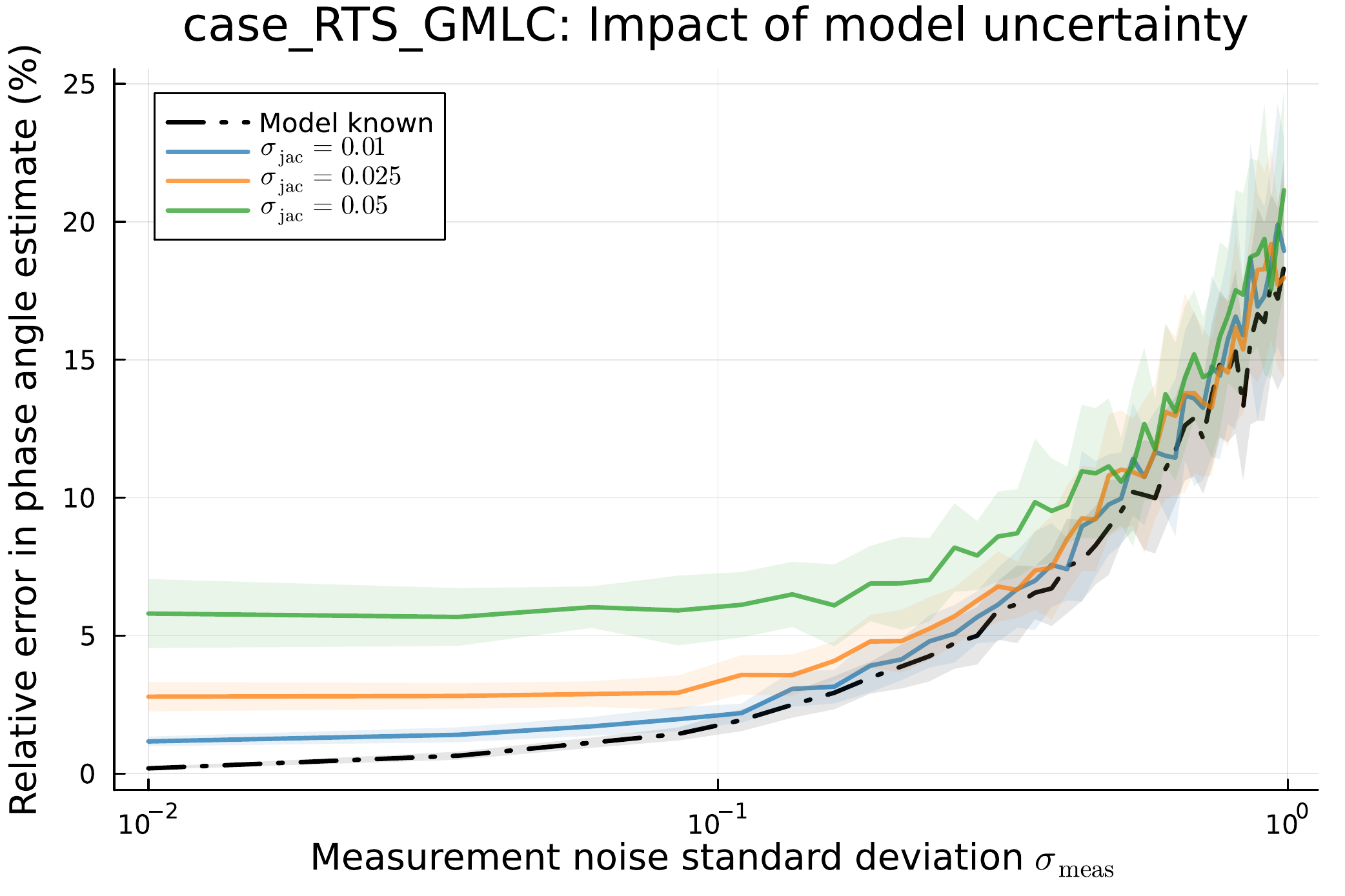}
            \caption{Impact of (\texttt{RTS\_GMLC}) model uncertainty on recovered phase angle relative error vs. measurement noise level. Shaded regions indicate $\pm$1 standard deviation of the relative errors computed over 20 bootstraps.}
            \label{fig:known-topology-comparison}
        \end{figure}

    \section{Discussion}
    \label{sec:discussion}
        The proposed extension of the phase retrieval framework to electric power systems has several potential consequences, key limitations, and directions for future work. Below, we discuss these items in Section \ref{sec:discussion-commentary}, Section \ref{sec:discussion-limitations}, and Section \ref{sec:future-work}, respectively.

        \subsection{Commentary on results}
        \label{sec:discussion-commentary}
  
        The proposed methods are a useful development from both a practical and theoretical perspective. These methods have important applications for numerous data-driven modeling scenarios in electric power systems where we only have access to measurements of the form \eqref{eq:ami_signal_definition}. The most obvious is in distribution systems, where it is typical to only have access to data like \eqref{eq:ami_signal_definition}. Applications in transmission systems are also appropriate, as large-scale PMU deployment remains an ongoing infrastructure upgrade challenge.

        The results in Lemma \ref{lemma:phaseless-symmetry} and its application in the program \eqref{eq:nr-phret-program-with-symmetry-constraints} are valuable results for data-driven power flow analysis. This is because they show that the state evolution matrices of the unmeasured voltage phase angle states can be recovered without knowledge of the voltage phase angles. The method can also be shown to be provably correct via the conditions based on spectral theory developed in Theorem \ref{thm:disc-delta-theta} and Theorem \ref{thm:disc}.

        \subsection{Limitations}
        \label{sec:discussion-limitations}
        The largest limitation of this research is that it primarily assumes a balanced power system that can be represented as a single-phase equivalent system.\footnote{Appendix \ref{apdx:unbalanced-extension} develops a preliminary extension to multiphase unbalanced networks.} An additional limitation is the fact that the phase retrieval program \eqref{eq:nr-phret-program-with-symmetry-constraints} relies on the structure of the power flow Jacobian, which is a time-varying matrix. While the matrix can be recursively estimated, e.g., as in \cite{xu_data-driven_2020,nowak-recursive-sensitivities-2020}, or via Kalman filtering; this is a key computational limitation.  A potential solution to this problem is to reformulate it in terms of the bus admittance matrix. Furthermore, the algorithm is only suitable for the PQ buses of the system or a system that only has PQ buses.

        \subsection{Future work}
        \label{sec:future-work}
        There are several key directions for future work to improve the practical application of the proposed methods. 

        \subsubsection{Multiphase extensions}
        One of the most clear research needs is generalizing the proposed algorithms to multiphase unbalanced networks. This can facilitate practical implementation beyond single-phase equivalent models in distribution systems and in other unbalanced systems. We provide preliminary results that indicate that the proposed algorithms can be generalized to multiphase unbalanced networks in Appendix \ref{apdx:unbalanced-extension}.

        \subsubsection{Additional phase retrieval formulations}
        Furthermore, recent advances in the phase retrieval literature can be synthesized with the problem setting described in this paper. For example, \cite{goldstein_phasemax_2018-1} and \cite{bahmani_phase_2017} develop convex formulations for the phase retrieval problem. These formulations may provide numerical and theoretical advantages in comparison to the approach provided by PhaseCut \cite{waldspurger_phase_2015}, which our work is largely based on. However, \cite{goldstein_phasemax_2018-1,bahmani_phase_2017} requires an appropriate guess for the phase vector, and the conditions for the effectiveness of the algorithms are based on the quality and construction of this initialization vector, which may be difficult to obtain in this application setting. An important direction for future work is to determine how other aspects of the phase retrieval framework \cite{jaganathan_phase_2016,DONG-JONATHAN-PHRET-REVIEW} can be applied to electric power systems and related to the AC power flow equations.

        \subsubsection{Partial/estimated topology knowledge}
        The limitation imposed by the time-varying nature of the power flow Jacobian could potentially be resolved by developing an extension of this work based on the nodal admittance matrix. This would require estimation of the network topology and parameters, e.g., as in \cite{yuan_inverse_2022,zhang_phaseless_topology_line_param,claeys_unbalanced_phaseless_line_param}. The nodal admittance matrix is, in principle, not time-varying. This could possibly allow measurements of the form \eqref{eq:ami_signal_definition} to be recovered without updating the power flow Jacobian.

        \subsubsection{Uncertain reactive power measurements}
        \label{sec:reactive-power-measurements}
       
        If reactive power measurements $\vq$ are unavailable, remotely readable reactive power control parameters\textemdash which are mandated to be available in the IEEE 1547-2018 standard \cite{noauthor_ieee_2018}\textemdash can be used to infer reactive power measurements when only the $\vv,\vp$ variables of \eqref{eq:ami_signal_definition} are available. Alternatively, distribution operators can use historical power factors and other heuristics, which are widely used in practice \cite{peppanen_distribution_2021}, to infer the reactive power injections. There is an opportunity for future work to develop methods to exploit knowledge of known/estimated reactive power control parameters to obtain this knowledge.

    \section{Conclusion}
    \label{sec:conclusion}
        This work proposed theory and algorithms to recover voltage phasors, current phasors, and power-voltage phase angle sensitivity matrices in electric power systems where the network model is unknown. The method requires only the voltage magnitudes, active power injections, and reactive power injections. The methods are directly inspired by the classical phase retrieval framework from the signal processing literature.

        To achieve these outcomes, this research developed a novel representation of the structure of the power flow Jacobian matrix (Lemma \ref{lemma:phaseless-symmetry}) that depends on neither the topology of the network nor the voltage phase angles. The proposed representation  of the structure of the power flow Jacobian was then used in Section \ref{sec:phasecut-nrpf} to provide a novel approach to recovering the voltage phase angles and a basis for an approximation of the power flow equations as a function of the phase angles. The method relies neither on partial observability of the phase angles nor on topological information. The proposed representation of the power flow Jacobian was also used to construct sufficient conditions (Theorem \ref{thm:disc-delta-theta}, Theorem \ref{thm:disc}) that certify when the method will work well. Consistent with the method itself, the conditions only require the measurements in \eqref{eq:ami_signal_definition}.

        Foreseeably, the proposed theory and algorithms could have significant practical engineering implications. The results of this research can enable engineers to ``fill in the information gap" in electric power networks that lack deployment of PMU measurements, as in \eqref{eq:pmu_signal_def}. It could also potentially save significant cost, by allowing fewer PMU measurements to be placed. As PMU deployment continues, these results can assist network control paradigms by retrieving comparable information to that provided by PMUs in areas that lack access to them.

    \section*{Acknowledgement}
     Samuel Talkington sincerely thanks Daniel Turizo for his feedback on the manuscript and valuable discussions.

\bibliographystyle{IEEEtran}
\bibliography{references.bib,bibs/phase_retrieval.bib,extras.bib,bibs/sensitivity.bib}

\newpage
\appendix
\section{Proof of Theorem 1}
\label{apdx:thm1-proof}
\begin{proof}
For any matrix $\mA \in \C^{n \times n}$, by the Gershgorin Circle Theorem the eigenvalues of $\mA$ are guaranteed to lie in the union of the $i=1,\dots,n$ Gershgorin discs $\setG_i(\mA)$ of the matrix, i.e.,
\begin{equation}
\label{eq:apdx:thm1:gershgorin-thm}
    \lambda_i(\mA) \in \bigcup_{i=1}^n \setG_i(\mA), \quad  i=1,\dots,n,
\end{equation}
where
\begin{equation}
    \setG_i(\mA) \triangleq \{w \in \C : | w - A_{ii}| \leq \sum_{k: k\neq i} |A_{ik}|\} \subseteq \C.
\end{equation}
It suffices to show that $\DPDTH$ is invertible to show that \eqref{eq:delta-theta-recovery-expression} holds. By \eqref{eq:apdx:thm1:gershgorin-thm}, the eigenvalues of $\DPDTH$ satisfy 
\begin{equation}
     \lambda_i\left(\DPDTH\right) \in \bigcup_{i \in \setB} \left\{w\in \C: \left|w - \frac{\partial p_i}{\partial \theta_i} \right| \leq \sum_{k \in \setB \setminus \{i\}} \left|\frac{\partial p_k}{\partial \theta_i}\right| \right\} \quad \forall i \in \setB.
\end{equation}
Therefore, if
\begin{equation}
\label{eq:thm1-step1}
    \forall i \in \setB, \quad \left|\frac{\partial p_i}{\partial \theta_i} \right| > \sum_{k \in \setB \setminus \{i\}} \left| \frac{\partial p_k}{\partial \theta_i}\right|, 
\end{equation}
then 0 fails to be an eigenvalue of $\DPDTH$, in which case, $\DPDTH$ is invertible and \eqref{eq:delta-theta-recovery-expression} holds. 

By Lemma \ref{lemma:phaseless-symmetry}, the symmetries in the entries of the power flow Jacobian \eqref{eq:thm1_expressions} hold, and we know that $\frac{\partial p_i}{\partial \theta_i} = v_i\frac{\partial q_i}{\partial v_i} - 2q_i$ for all $i \in \setB$, and $\frac{\partial p_i}{\partial \theta_k} = v_k \frac{\partial q_i}{\partial v_k}$ for all $k \in \setB\setminus\{i\}$. 

Applying the power flow Jacobian structure \eqref{eq:thm1_expressions} to \eqref{eq:thm1-step1}, we obtain 
    \begin{equation}
    \label{eq:thm1-step2}
        \forall i \in \setB \quad |v_i| \sum_{k \in \setB \setminus \{i\}} \left| \frac{\partial q_k}{\partial v_i}\right| < \left|v_i\frac{\partial q_i}{\partial v_i} -2q_i \right|.
    \end{equation}
Moreover, by applying the triangle inequality to \eqref{eq:thm1-step2}, we have
\begin{equation}
\label{eq:thm1-penultimate-row-offdiag}
    \forall i \in \setB \quad |v_i| \sum_{k \in \setB \setminus \{i\}} \left| \frac{\partial q_k}{\partial v_i}\right|  < \left|v_i\frac{\partial q_i}{\partial v_i} \right| + 2 \left|q_i \right|.
\end{equation}
By construction $v_i>0$, so $|v_i| = v_i$. Rearranging \eqref{eq:thm1-penultimate-row-offdiag}, if 
\begin{equation}
\label{eq:thm1-apdx-row-condition}
    \forall i \in \setB \quad |q_i| > \frac{1}{2} v_i \left( \sum_{k \in \setB \setminus \{i\}}
    \left| \frac{\partial q_k}{\partial v_i}\right| - \left| \frac{\partial q_i}{\partial v_i}\right|
    \right),
\end{equation}
then 0 fails to be an eigenvalue of $\DPDTH$. This means that $\DPDTH$ is invertible, and therefore, it is guaranteed that we can uniquely recover the phase angle deviations $\Delta \vtheta$ as \eqref{eq:delta-theta-recovery-expression}.

This completes the proof for the row off-diagonal condition in \eqref{eq:thm1_row_offdiag}. The column off-diagonal condition \eqref{eq:thm1_column_offdiag} can be shown through the same application of the power flow Jacobian symmetry and the triangle inequality, and rearranging appropriately. We have that
\begin{subequations}
\label{eq:apdx:thm1-column-condition-inequalities}
\begin{align}
 \forall i \in \setB \quad  \left| \frac{\partial p_i}{\partial \theta_i}\right| &= \left| v_i \frac{\partial q_i}{\partial v_i}  + 2 q_i\right|\\
 &< v_i \left| \frac{\partial q_i}{\partial v_i} \right| + 2 |q_i|\\
 &> \sum_{k \in \setB \setminus \{i\}}\left| v_k \frac{\partial q_i}{\partial v_k} \right| =  \sum_{k \in \setB \setminus \{i\}} v_k\left|  \frac{\partial q_i}{\partial v_k} \right|.
\end{align}
\end{subequations}
Rewriting \eqref{eq:apdx:thm1-column-condition-inequalities}, we obtain the desired condition that if 
\begin{equation}
\label{eq:thm1-apdx-column-condition}
      \forall i \in \setB \quad |q_i| > \frac{1}{2}\left(\sum_{k \in \setB \setminus \{i\}} v_k \left| \frac{\partial q_i}{\partial v_k}\right| - v_i \left| \frac{\partial q_i}{\partial v_i} \right| \right),
\end{equation}
then 0 fails to be an eigenvalue of $\DPDTH$, and we can uniquely recover the phase angle perturbations $\Delta \vtheta$ as \eqref{eq:delta-theta-recovery-expression}. 

Note that if either the row off-diagonal bound \eqref{eq:thm1-apdx-row-condition} or the column off-diagonal bound \eqref{eq:thm1-apdx-column-condition} holds for all $i \in \setB$, then 0 is guaranteed to not be an eigenvalue of $\DPDTH$, which means that $\DPDTH$ is guaranteed to be invertible. This means that the recovery of $\Delta \vtheta$ via \eqref{eq:delta-theta-recovery-expression} is guaranteed to exist and be unique.

Identical logic can be applied to yield the bounds for the active power injection \eqref{eq:thm1:p-conditions}. By Lemma \ref{lemma:phaseless-symmetry} and the triangle inequality, if for all $i \in \setB$
\begin{subequations}
    \begin{align}
        \left|\frac{\partial q_i}{\partial \theta_i} \right| &= \left| -v_i \frac{\partial p_i}{\partial v_i} + 2 p_i\right|,\\
        &\leq v_i \left| \frac{\partial p_i}{\partial v_i}  \right| + 2 |p_i|,\\
        &> \sum_{k \in \setB \setminus\{i\}} \left| \frac{\partial q_k}{\partial \theta_i}\right|,\\
        &= |v_i| \sum_{k \in \setB \setminus\{i\}} \left| \frac{\partial p_k}{\partial v_i}\right| = v_i \sum_{k \in \setB \setminus\{i\}} \left| \frac{\partial p_k}{\partial \theta_i}\right|,
    \end{align}
\end{subequations}
or the corresponding column off-diagonal condition holds, then 0 is not an eigenvalue of $\DQDTH$, and there exists unique phase angle perturbations such that \eqref{eq:delta-theta-reactive-recovery-expression} holds. 
\end{proof}

\section{Proof of Theorem 2}
\label{apdx:thm2-proof}

\begin{table}[b]
    \centering
    \begin{tabular}{|c|c|}
    \hline
         Case & Assumption \ref{assum:full-rank-diagonal-blocks} holds?  \\
         \hline
         \texttt{14} & \textcolor{olive}{Yes}\\ 
         \texttt{24\_ieee\_rts} &  \textcolor{olive}{Yes}\\
         \texttt{ieee30} &  \textcolor{olive}{Yes}\\
         \texttt{RTS\_GMLC} & \textcolor{olive}{Yes}\\
         \texttt{118} & \textcolor{olive}{Yes}\\
        \texttt{89pegase} & \textcolor{olive}{Yes}\\
         \texttt{ACTIVSg200} & \textcolor{olive}{Yes}\\
         \texttt{ACTIVSg500} & \textcolor{olive}{Yes}\\
         \texttt{ACTIVSg2000} & \textcolor{olive}{Yes} \\
         \hline
    \end{tabular}
    \caption{All test networks satisfy Assumption \ref{assum:full-rank-diagonal-blocks}.}
    \label{tab:assum}
\end{table}

\begin{proof}
The Gershgorin Circle Theorem can be generalized to block matrices, as shown in \cite[Corollary 3.2]{echeverria_block_2018} and \cite[Definition 3]{feingold_block_1962}. Consider a block matrix $\mM$ of the form
\begin{equation}
    \mM = [\mA_{ij}] \quad \textsf{with blocks} \quad \mA_{ij} \in \C^{m \times m}, \quad i,j=1,\dots,n,
\end{equation}
where $\rank(\mA_{ii}) = m$ for all $i=1,\dots,n$. According to \cite[Corollary 3.2]{echeverria_block_2018}, if $\lambda \in \C$ is an eigenvalue of $\mM$, then there exists  at least one  $i \in \{1,\dots,n\}$ such that
\begin{equation}
\label{eq:echeverria-block-gershgorin}
    \sum_{j \neq i}^n \norm{(\mA_{ii} - \lambda \mId)^{-1} \mA_{ij}}_2 \geq 1,
\end{equation}
where $\norm{\mA}_2$ is the matrix 2-norm, or spectral norm, which is defined as $\norm{\mA}_2 \triangleq \max_{\vx \neq \boldsymbol{0}} \frac{\norm{\mA \vx}}{\norm{\vx}_2} = \sigma_{\sf max}(\mA)$, where $\sigma_{\sf max}(\mA)$ is the largest singular value of $\mA$.

Furthermore, every eigenvalue $\lambda$ of $\mM$ lies in the unions of these discs.
\begin{equation}
    \lambda(\mM) \in \bigcup_{i =1}^n \left\{ w \in \C : \sum_{j\neq i}^n \norm{(\mA_{ii} - w \mId)^{-1} \mA_{ij}}_2 \geq 1\right\}.
\end{equation}
If all the blocks $\mA_{ij}$ have dimension $1 \times 1$, then this statement reduces to the classical Gershgorin Circle Theorem. Suppose that $\mM$ takes the form of a $2n \times 2n$ block matrix that satisfies Assumption \ref{assum:full-rank-diagonal-blocks}, which we write as 
\begin{equation}
\renewcommand*{\arraystretch}{1.5}
    \mM \triangleq \begin{bmatrix}
        \mA_{11} & \mA_{12}\\
        \mA_{21} & \mA_{22}
    \end{bmatrix},
\end{equation}
where $\rank(\mA_{11}) = \rank(\mA_{22}) = n$.
Then, applying \eqref{eq:echeverria-block-gershgorin}, if 0 is an eigenvalue of $\mM$, then
\begin{equation}
    \norm{\mA_{11}^{-1} \mA_{12}}_2 \geq 1  \quad \textsf{or}  \quad \norm{\mA_{22}^{-1} \mA_{21}}_2 \geq 1.
\end{equation}
The contraposition is that if
\begin{equation}
\label{eq:apdx:thm2:blockgershgorinsufficient}
    \norm{\mA_{11}^{-1} \mA_{12}}_2 < 1  \quad \textsf{and}  \quad \norm{\mA_{22}^{-1} \mA_{21}}_2 <1,
\end{equation}
then 0 is guaranteed to \emph{not} be an eigenvalue of $\mM$. That is, \eqref{eq:apdx:thm2:blockgershgorinsufficient} is sufficient to certify the invertibility of the block matrix $\mM$. 

We can complete the proof by applying the above logic to the power flow Jacobian. Let $\mJ \in \R^{2|\setB|\times 2|\setB|}$ be the power flow Jacobian for PQ buses $\setB$ under study. Recall that we defined $\vx = [\vtheta^T, \vv^T]^T \in \R^{2|\setB|}$ as a guess for the network state. We want to establish when the entire block matrix
\begin{equation}
\renewcommand*{\arraystretch}{1.5}
       \mJ(\vx) \triangleq \begin{bmatrix}
       \DPDTH(\vx) & \DPDV(\vx)\\
            \DQDTH(\vx) & \DQDV(\vx)
       \end{bmatrix} 
    \end{equation}
    is invertible\textemdash with $\vtheta$ unknown. This is equivalent to showing that $0$ is not an eigenvalue of $\mJ$. By Assumption \ref{assum:full-rank-diagonal-blocks}, $\DPDTH$ and $\DQDV$ are invertible\footnote{Assumption \ref{assum:full-rank-diagonal-blocks} empirically holds for all test cases studied, including large synthetic real-world networks, as shown in Table \ref{tab:assum}.}. Then by \cite[Corollary 3.2]{echeverria_block_2018}, if $\lambda \in \C$ is an eigenvalue of $\mJ$, then either
\begin{equation}
\norm{\left(\DPDTH - \lambda \mId\right)^{-1} \DPDV}_2 \geq 1
\end{equation}
or
\begin{equation}
    \norm{\left(\DQDV - \lambda \mId\right)^{-1} \DQDTH}_2 \geq 1.
\end{equation}
Let $\mJ_{\#}: \R^{|\setB|} \times \R^{|\setB|} \times \R^{|\setB|} \mapsto \R^{2 |\setB| \times 2 |\setB|}$ be a matrix-valued function of the measurements in \eqref{eq:ami_signal_definition}. Consider the power-voltage phase angle sensitivity matrix functions $\DPDTH(\cdot,\cdot),\DQDTH(\cdot,\cdot): \R^{|\setB|} \times \R^{|\setB|} \mapsto \R^{|\setB| \times |\setB|}$ as defined in \eqref{eq:augmented-phase-sensitivities}. By parameterizing $\mJ_{\#}$ by $\vtheta$, we can write
    \begin{equation}
    \renewcommand*{\arraystretch}{1.5}
        \mJ_{\#}(\vv,\vp,\vq | \vtheta)  \triangleq \begin{bmatrix}
       \DPDTH(\vv,\vq) & \DPDV(\vv|\vtheta)\\
            \DQDTH (\vv,\vp) & \DQDV(\vv|\vtheta)
       \end{bmatrix},
    \end{equation}
    which is a function of \eqref{eq:ami_signal_definition}. Then, by Lemma \ref{lemma:phaseless-symmetry}, we have $\lambda \in \lambda(\mJ) \iff \lambda \in \lambda( \mJ_{\#})$. 
    
    Therefore, it suffices to show that 0 is not an eigenvalue of $ \mJ_{\#}$ to complete the proof.  Applying \cite[Corollary 3.2]{echeverria_block_2018} again, if $\lambda$ is an eigenvalue of $\mJ_{\#}$, then either
\begin{equation}
\norm{\left(\DPDTH(\vv,\vq) - \lambda \mId\right)^{-1} \DPDV}_2 \geq 1
\end{equation}
or
\begin{equation}
    \norm{\left(\DQDV - \lambda \mId\right)^{-1} \DQDTH(\vv,\vp)}_2 \geq 1.
\end{equation}
     Finally, applying Lemma \ref{lemma:phaseless-symmetry} and forming the contrapositive, if
    \begin{equation}
    \label{eq:apdx2:proofthm2:final}
        \norm{\left(\DPDTH(\vv,\vq)\right)^{-1} \DPDV}_2 < 1, \quad \textsf{and} \quad  \norm{\left(\DQDV\right)^{-1} \DQDTH(\vv,\vq)}_2 <1,
    \end{equation}
    then 0 fails to be an eigenvalue of $\mJ$. Note that the inequalities \eqref{eq:apdx2:proofthm2:final} depend solely on \eqref{eq:ami_signal_definition}. This completes the proof for \eqref{eq:thm2_row_offdiag}. Analogous arguments can be applied to $\mJ_{\#}^T$ to yield the second set of inequalities \eqref{eq:thm2_column_offdiag}, because $\sigma(\mJ) = \sigma(\mJ^T)$ and $\lambda(\mJ) = \lambda(\mJ^T)$.


    As in Theorem \ref{thm:disc-delta-theta}, conditions \eqref{eq:thm2_row_offdiag} and \eqref{eq:thm2_column_offdiag} allow us to certify Theorem \ref{thm:disc} for a subset of PQ buses $\setB$ under study. This can be achieved by exploiting the properties of singular values of principle submatrices. Let $\mS_{\bullet} \in \R^{(|\setB| -b) \times (|\setB| -b)}$ be a principal submatrix of one of the Jacobian submatrix products, 
    \begin{equation}
    \label{eq:jac-submatrix-products}
        \mS \in \left\{ \DPDTH(\vv,\vq)^{-1} \DPDV, \ \ \DQDV^{-1} \DQDTH(\vv,\vp), \ \ \DPDTH(\vv,\vq)^{-1}\DQDTH(\vv,\vp), \ \ \DQDV^{-1} \DPDV \right\},
    \end{equation}
    with the corresponding rows and columns of $b < |\setB|$ buses removed. Note that $\mS$ is one of the arguments of $\norm{\cdot}_2$ for the spectral norm  conditions in \eqref{eq:thm2_row_offdiag} or \eqref{eq:thm2_column_offdiag}. Then, by the submatrix property of singular values \cite[\S 1.8]{strang_linear_2019}, $\forall i=1,\dots,|\setB| -b,$ the singular values of $\mS_{\bullet} $ satisfy $\sigma_{i + 2b}(\mS) \leq \sigma_i(\mS_{\bullet}) \leq \sigma_i(\mS)$ for all Jacobian submatrix products $\mS$ as defined in \eqref{eq:jac-submatrix-products}. Thus, the conditions \eqref{eq:thm2_row_offdiag} and \eqref{eq:thm2_column_offdiag} remain sufficient for the $|\setB|-b$ buses.

\end{proof}

\section{Additional numerical results}
\label{apdx:additional-numres}

\subsection{Extension to three-phase unbalanced networks}
        \label{apdx:unbalanced-extension}
        Consider a three-phase unbalanced distribution network with $n$ buses. The \emph{compound admittance matrix} of the network, $\mY \in \C^{3n \times 3n}$, is a block matrix $\mY= \left[\mM_{i,j}\right]_{i,j=1}^n$. Block $\mM_{i,j} \in \C^{3 \times 3}$ describes the mutual admittance between the three phases of branch $i \to j$ with entries $\mM_{i,j} = -\big[y^{\phi}_{i,j}\big]_{\phi \in \setF_{ij}}$. Set $\setF_{ij}$ denotes the phases of the branch $i \to j$.  The diagonal blocks $\mM_{i,i}$ are computed as $\mM_{i,i} = \mY_i -\sum_{j:j\neq i} \mM_{i,j}$, where $\mY_i \in \C^{3 \times 3}$ is the equivalent admittance of the loads connected at bus $i$, and the admittance-to-ground. We refer the reader to \cite{bazrafshan_comprehensive_2018} for more comprehensive information on modeling three-phase networks with the bus admittance matrix.

        As a proof-of-concept to verify that the proposed symmetry in Lemma \ref{lemma:phaseless-symmetry} and the phase retrieval method \eqref{eq:nr-phret-program-with-symmetry-constraints} is applicable in the unbalanced case, we compute the power-voltage phase angle sensitivities using the compound admittance matrix $\mY$. According to \cite[5.2.1]{molzahn_survey_2019}, we can express the complex-power to voltage magnitude and complex power-voltage phase angle sensitivities as complex matrices $\DSDV,\DSDTH \in \C^{n \times n}$ of the form
        \begin{align}
             \DSDV &= \diag(\vvbar)\left( \diag(\conj{\mY} \conj{\vvbar}) + \conj{\mY} \diag(\conj{\vvbar})\right) \diag(\vv)^{-1},\\
            \DSDTH &= j \diag(\vvbar)\left( \diag(\conj{\mY} \conj{\vvbar}) - \conj{\mY} \diag(\conj{\vvbar})\right),
        \end{align}
        where $\vvbar \in \C^n$ are the complex voltage phasors, $\vv = |\vvbar|$ are the voltage magnitudes, and $\conj{(\cdot)}$ denotes the complex conjugate. We then have $\DPDV,\DQDV = \Re{\DSDV},\Im{\DSDV}$, and $\DPDTH,\DQDTH=\Re{\DSDTH},\Im{\DSDTH}$. 
        
        These closed-form expressions allow us to verify that the structural symmetries derived in Lemma \ref{lemma:phaseless-symmetry} hold in the unbalanced case. In Table \ref{tab:muti-phase-works}, we show the relative error between the deterministic $\DPDTH,\DQDTH$ sensitivities and the model-free power-voltage phase angle sensitivity functions $\DPDTH(\vv,\vq),\DQDTH(\vv,\vp)$ as defined in  \eqref{eq:augmented-phase-sensitivities}. The results are shown for a three-phase unbalanced distribution test network with three buses.
        
        Data used to compute the values in Table \ref{tab:muti-phase-works} were provided by \texttt{PowerModelsDistribution.jl} \cite{powermodels_distribution}. The case used was the three-phase unbalanced \texttt{case\_3\_unbalanced.dss}. The \texttt{ACPUPowerModel} polar AC power flow model was used to obtain a solution. The voltage magnitudes, active powers, and reactive powers were then computed at the AC power flow solution, and used to generate $\DPDTH(\vv,\vq)$ and $\DQDTH(\vv,\vp)$ as defined in  \eqref{eq:augmented-phase-sensitivities}.

\begin{table}[t]
    \centering
    \begin{tabular}{|c|c|}
    \hline
        Quantity & Value\\
        \hline
        $\norm{\DPDTH -\DPDTH(\vv,\vq)}_F \Big/ \norm{\DPDTH}_F$  & $2.510 \times 10^{-8}$ \\
        $\norm{\DQDTH - \DQDTH(\vv,\vp)}_F \Big/ \norm{\DQDTH}_F$ & $1.725 \times 10^{-7}$ \\
        \hline
    \end{tabular}
    \caption{Verification that the structure expressions \eqref{eq:thm1_expressions} of Lemma \ref{lemma:phaseless-symmetry} hold for multiphase unbalanced networks.}
    \label{tab:muti-phase-works}
\end{table}

\subsection{Verifying Lemma \ref{lemma:phaseless-symmetry} in deterministic Newton-Raphson power flow}
    \label{apdx:phaseless-nrpf-iter}
               \begin{table}[t]
        \renewcommand{\arraystretch}{1.2}
            \centering
            \begin{tabular}{|c||c|c|c|c|}
            \hline
               \multirow{2}{*}{Iteration} &  \multicolumn{4}{c|}{ Relative error $\frac{\|(\cdot) - \hat{(\cdot)}\|}{\|(\cdot)\|} \times 10^{16}$  }\\ \cline{2-5} 
                & $\vtheta$  & $\vell $ & $\DPDTH$ & $\DQDTH$  \\
                \hline
                 t=1 & $4.922 $ & $1.513 $ & $2.407 $ & 5.096 \\
                 t=2 & $5.324 $ & 1.907 &  2.525 & 3.381  \\
                 t=3 & $7.444 $ & 1.979 & 3.020 & 5.118 \\
                 t=4 & $4.377 $ & 1.514  & 2.066 & 3.001 \\
                 t=5 & $4.681 $ & 1.497 & 2.904 &  6.780 \\
                 t=6 & $3.530 $ & 1.182 & 2.973 & 5.583 \\
                 \hline
            \end{tabular}
            \caption{Relative error ($\times 10^{16}$) between the phaseless and deterministic Newton-Raphson power flow model with no noise for the \texttt{RTS\_GMLC} test case \cite{barrows_ieee_2019}.}
            \label{tab:iterative_rts_gmlc}
        \end{table}

        The phaseless sensitivity functions $\DPDTH(\vv,\vq),\DQDTH(\vv,\vp)$ as defined in \eqref{eq:augmented-phase-sensitivities} are equivalent to the deterministic matrices $\DPDTH,\DQDTH$ in \eqref{eq:dpdth_jac_block}, \eqref{eq:dqdth_jac_block}. This can be shown by using $\DPDTH(\vv,\vq),\DQDTH(\vv,\vp)$ to iteratively solve the Newton-Raphson power flow model \eqref{eq:nr_pf_definition} in the noiseless case without considering $\vtheta$ as a grid state. To this end, we compare the quantities $\DPDTH,\DQDTH,$ $\vtheta,$ and $\vell$ at every Newton-Raphson iteration $t=1,\dots,$ when solving the \texttt{RTS\_GMLC} test case with both the phaseless and deterministic $\DPDTH,\DQDTH$ in Table \ref{tab:iterative_rts_gmlc}. The results  indicate that\textemdash with no observations of the phase angles\textemdash comparative results for the Newton-Raphson power flow model are achieved \emph{within machine precision} at each iteration. This demonstrates that the results in Lemma \ref{lemma:phaseless-symmetry} are \emph{physically} valid.
        
 

    \begin{figure}[t]
        \centering
        \includegraphics[width=0.49\linewidth,keepaspectratio]{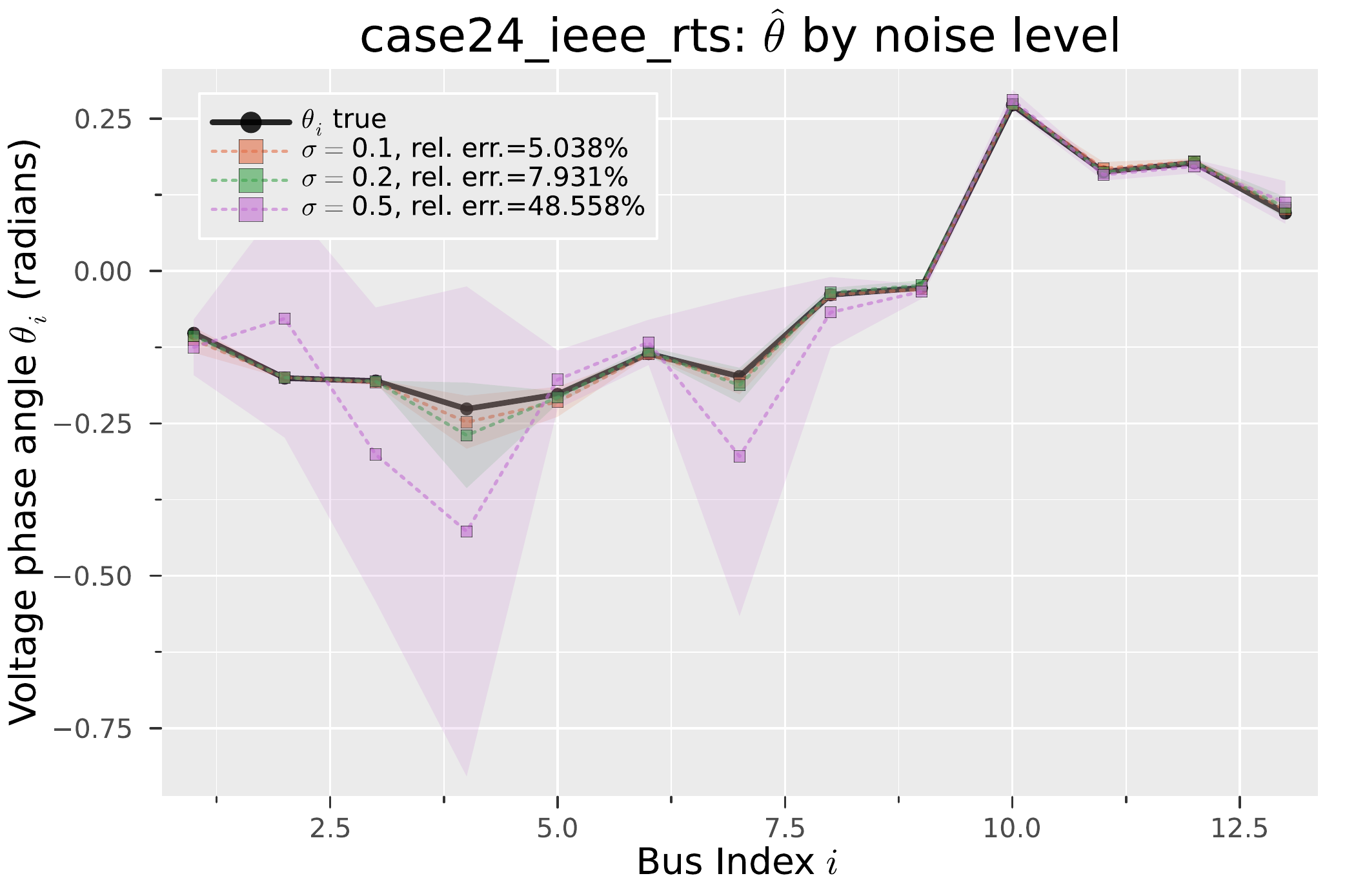}\hfill
        \includegraphics[width=0.49\linewidth,keepaspectratio]
        {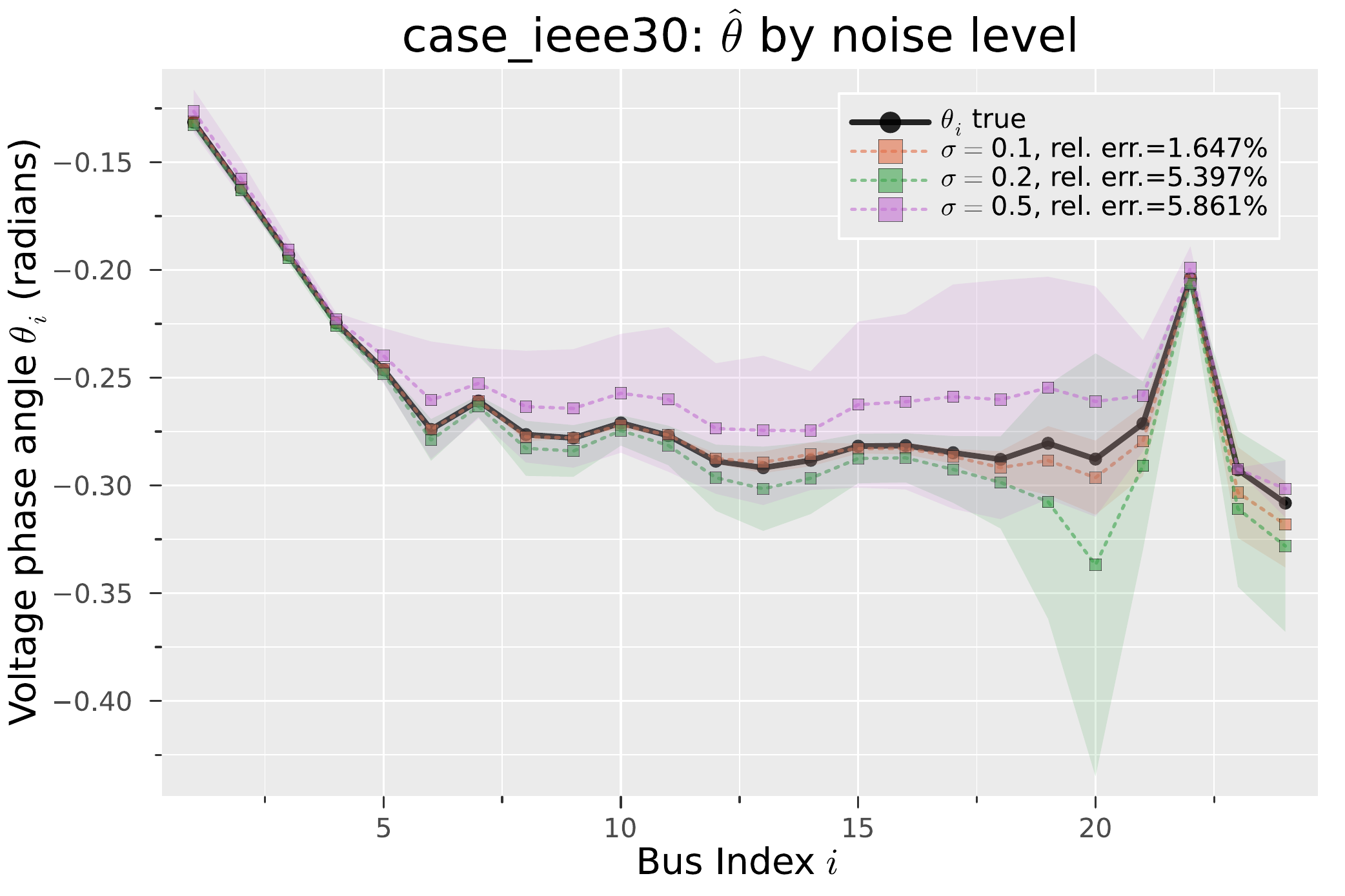}
        \caption{Performance of \eqref{eq:nr-phret-program-with-symmetry-constraints} by measurement noise level. The shaded ribbon indicates a $\pm 2 e$ error interval where $e$ is the absolute bus voltage phase error.}
        \label{fig:apdx:phret_results_by_noise}
    \end{figure}
    
    \subsection{Recovery of voltage phase angles and power-voltage phase angle sensitivities}
    \label{apdx:recovery-voltage-phase-angles}
    Additional numerical simulations are shown in this appendix for the same phase retrieval technique \eqref{eq:nr-phret-program-with-symmetry-constraints} used for the simulations in Section \ref{sec:numerical-ph-angle-recovery}. In particular, we report results for the IEEE 30-bus test system and the IEEE 24-bus reliability test system (RTS).

    In Fig. \ref{fig:apdx:phret_results_by_noise}, we show the estimation of the bus voltage phase angles at the AC power flow solution using \eqref{eq:nr-phret-program-with-symmetry-constraints} for additional test cases with multiple noise levels. The experiments are the same as was done in Section  \ref{sec:numerical-ph-angle-recovery}, which yielded the prior results in Fig. \ref{fig:phret_results_by_noise} for the \texttt{RTS\_GMLC} network model.   Fig. \ref{fig:apdx:sensitivity-results} shows the power-voltage phase angle sensitivity recovery results for additional test cases, in the same manner as was done in Section  \ref{sec:sensitivity-results}, which yielded the prior results in Fig. \ref{fig:sensitivity-results} for the \texttt{RTS\_GMLC} network model.
            \begin{figure}[t]
            \centering
            \includegraphics[width=0.49\linewidth,keepaspectratio]{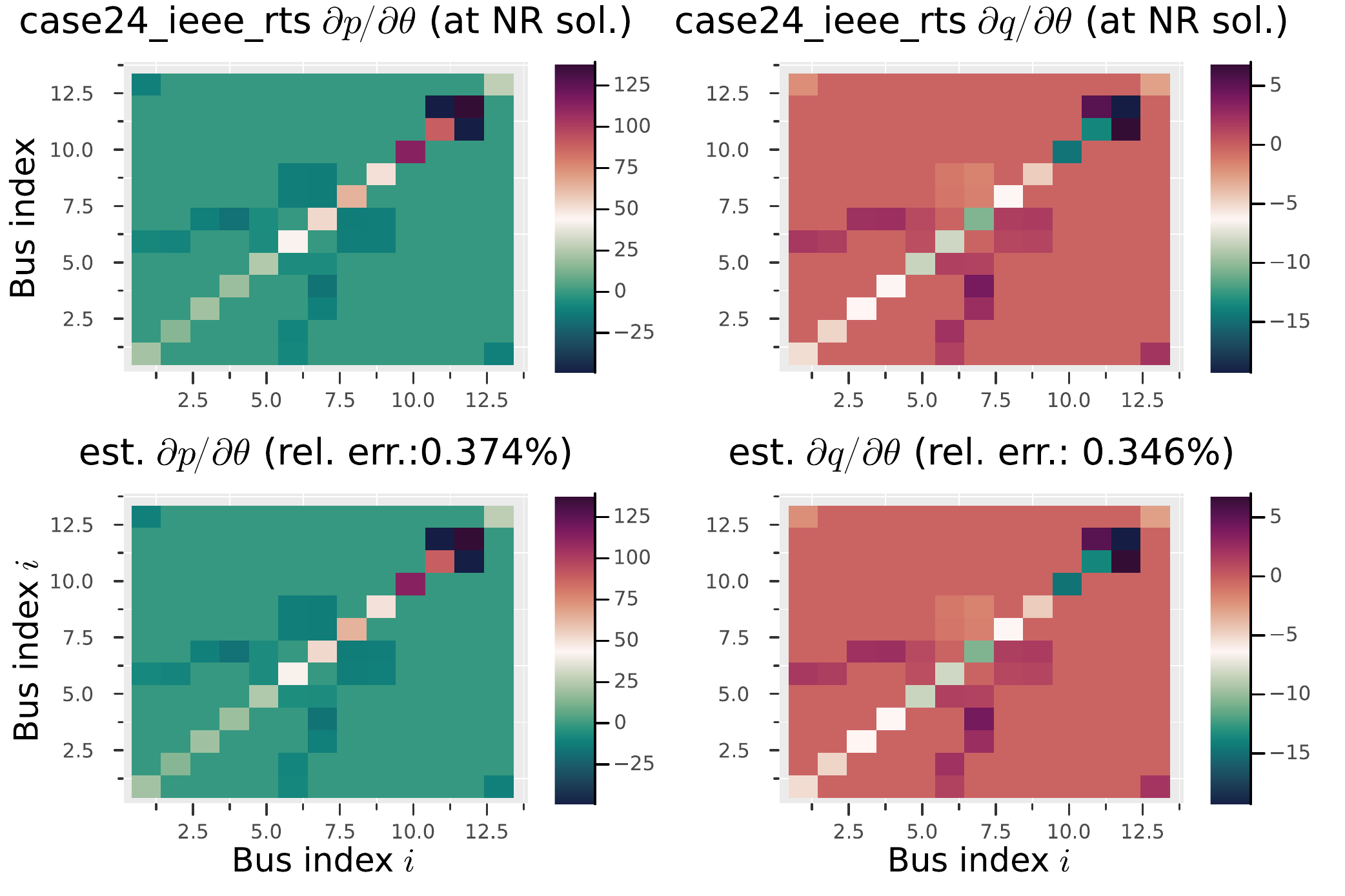}
            \includegraphics[width=0.49\linewidth,keepaspectratio]{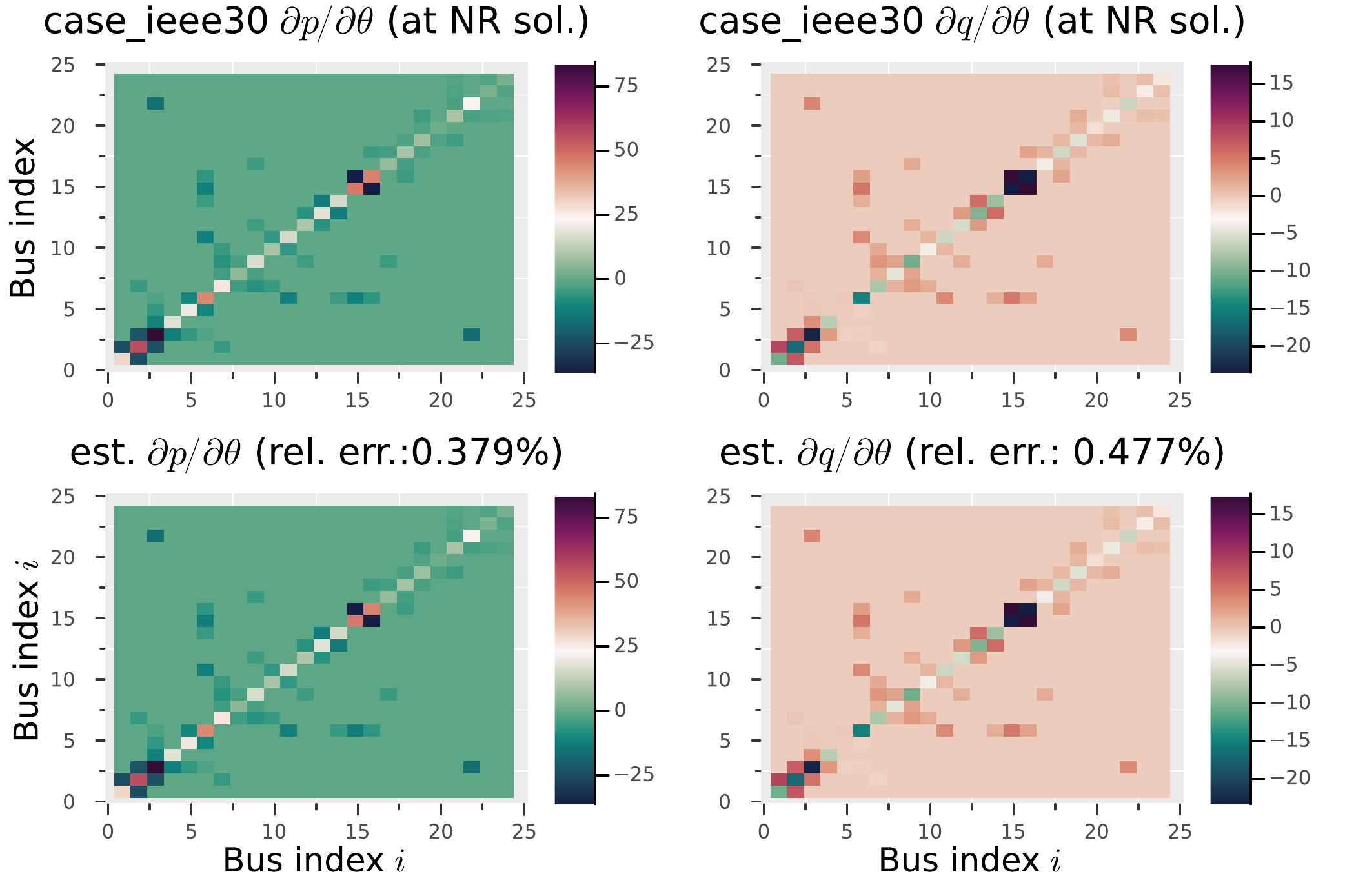}\hfill
            \hfill
            \caption{Recovery of the power-voltage phase angle submatrices of the AC power flow Jacobian for additional test cases via the 
            phase retrieval program \eqref{eq:nr-phret-program-with-symmetry-constraints} with $\sigma_{\sf meas} = 0.1$. Bus indices correspond to PQ bus ordering of the basic network format of \texttt{PowerModels.jl} \cite{powermodels}. Frobenius norm relative error is displayed on the second row.}
            \label{fig:apdx:sensitivity-results}
        \end{figure}
    


\end{document}

%% file: head.tex
\usepackage[export]{adjustbox} 
\usepackage{amsmath} 
\usepackage{amssymb}  
\usepackage{mathtools}
\usepackage[english]{babel}
\usepackage{amsfonts}
\usepackage{ dsfont }
\usepackage{wrapfig}
\usepackage{graphicx}

\usepackage{graphicx, multicol} 
\usepackage{xcolor} 
\usepackage{marvosym, wasysym} 
\usepackage{listings}
\definecolor{codegreen}{rgb}{0,0.6,0}
\definecolor{codegray}{rgb}{0.5,0.5,0.5}
\definecolor{codepurple}{rgb}{0.58,0,0.82}
\definecolor{backcolour}{rgb}{0.95,0.95,0.92}
\lstdefinestyle{mystyle}{
    backgroundcolor=\color{backcolour},   
    commentstyle=\color{codegreen},
    keywordstyle=\color{magenta},
    numberstyle=\tiny\color{codegray},
    stringstyle=\color{codepurple},
    basicstyle=\ttfamily\footnotesize,
    breakatwhitespace=false,         
    breaklines=true,                 
    captionpos=b,                    
    keepspaces=true,                 
    numbers=left,                    
    numbersep=5pt,                  
    showspaces=false,                
    showstringspaces=false,
    showtabs=false,                  
    tabsize=2
}

%% file: style.tex
\usepackage{hyperref}
\hypersetup{
    colorlinks=false,
    linkcolor=blue,
    filecolor=magenta,      
    urlcolor=cyan,
    pdfpagemode=FullScreen,
    }

\usepackage{geometry} 
\usepackage{graphicx} 

\usepackage{booktabs} 
\usepackage{array} 
\usepackage{paralist} 
\usepackage{verbatim} 
\usepackage{subfig} 
\usepackage{steinmetz}

\usepackage{fancyhdr} 
\usepackage{sectsty}
\allsectionsfont{\sffamily\mdseries\upshape} 

\usepackage[nottoc,notlof,notlot]{tocbibind} 
\usepackage[titles,subfigure]{tocloft} 

\usepackage{cite}

%% file: defs-math-power.tex
\input{defs-math}

\newcommand{\DPDTH}{\frac{\partial \vp}{\partial \vtheta}}
\newcommand{\DQDV}{\frac{\partial \vq}{\partial \vv}}
\newcommand{\DQDTH}{\frac{\partial \vq}{\partial \vtheta}}
\newcommand{\DPDV}{\frac{\partial \vp}{\partial \vv}}

\newcommand{\DSDV}{\frac{\partial \vs}{\partial \vv}}
\newcommand{\DSDTH}{\frac{\partial \vs}{\partial \vtheta}}

\newcommand{\diag}{\operatorname{diag}}



%% file: defs-math.tex
\usepackage{amsmath,amssymb}

\newcommand{\R}{\mathbb{R}}

\newcommand{\C}{\mathbb{C}}

\renewcommand{\Re}[1]{\operatorname{Re}\left\{#1\right\}}
\renewcommand{\Im}[1]{\operatorname{Im}\left\{#1\right\}}
\newcommand{\conj}[1]{\mkern 1.5mu\underline{\mkern-1.5mu#1\mkern-1.5mu}\mkern 1.5mu}

\newcommand{\vct}[1]{\boldsymbol{#1}}
\newcommand{\mtx}[1]{\boldsymbol{#1}}

\newcommand{\trace}[1]{\operatorname{trc}\left(#1 \right)}
\newcommand{\rank}{\operatorname{rank}}

\newcommand{\norm}[1]{\left|\left|#1\right|\right|}

\newcommand{\set}[1]{\mathcal{#1}}

\DeclareMathOperator*{\minimize}{\text{minimize}}

\newcommand{\va}{\vct{a}}
\newcommand{\vb}{\vct{b}}

\newcommand{\vd}{\vct{d}}

\newcommand{\vf}{\vct{f}}
\newcommand{\vg}{\vct{g}}

\newcommand{\vp}{\vct{p}}
\newcommand{\vq}{\vct{q}}

\newcommand{\vs}{\vct{s}}

\newcommand{\vu}{\vct{u}}
\newcommand{\vv}{\vct{v}}

\newcommand{\vx}{\vct{x}}
\newcommand{\vy}{\vct{y}}

\newcommand{\vtheta}{\vct{\theta}}

\newcommand{\vxi}{\vct{\xi}}

\newcommand{\vell}{\vct{\ell}}

\newcommand{\mA}{\mtx{A}}
\newcommand{\mB}{\mtx{B}}

\newcommand{\mG}{\mtx{G}}

\newcommand{\mJ}{\mtx{J}}

\newcommand{\mM}{\mtx{M}}

\newcommand{\mS}{\mtx{S}}

\newcommand{\mX}{\mtx{X}}
\newcommand{\mY}{\mtx{Y}}

\newcommand{\mXi}{\mtx{\Xi}}

\newcommand{\mId}{\mathbf{I}}

\newcommand{\vvbar}{\overline{\vct{v}}}

\newcommand{\setA}{\set{A}}
\newcommand{\setB}{\set{B}}

\newcommand{\setF}{\set{F}}
\newcommand{\setG}{\set{G}}

\newcommand{\setN}{\set{N}}


%% file: figs/circuit.tex
\newcommand{\bushere}[2]{%
    coordinate(tmp)
    ++(0,1) node[above]{#1} edge[ultra thick] ++(0,-2)
    ++(0,-2) node[below]{#2}
    (tmp)
}
\begin{circuitikz}[thick]
    \draw (0,0) node[vsourcesinshape, rotate=90](S){} (S.east) 
        -- ++(0,2.5)
        node[above]{} -- ++(1.5,0)
        \bushere{$ 1.0\angle 0$}{Bus 1} 
        -- ++(1,0)
        to[generic, l_={$j0.1 \ \textrm{p.u.}$}] ++(3,0) 
        -- ++(1,0)
        \bushere{$v \angle\theta$}{Bus 2} -- ++(0.5,0)
        node[]{} -- ++(1,0)
        edge[-Stealth] ++(0,-3)
        ++(-0.5,-3.25) node[]{$2+j1 \ \textrm{p.u.}$}
        ;
\end{circuitikz}